\documentclass[11pt,letterpaper]{article}

\usepackage{amsthm,amsmath,amssymb,color,fullpage,cleveref}
\usepackage{mathrsfs}
\usepackage[margin=1in]{geometry}
\usepackage[singlespacing]{setspace}
\usepackage[ruled,vlined]{algorithm2e}

\usepackage{amsfonts}
\usepackage{dsfont}

\theoremstyle{definition}

\usepackage[utf8]{inputenc} 
\usepackage[T1]{fontenc}
\usepackage{graphicx,colordvi,psfrag}
\usepackage{amsmath,amssymb}
\usepackage{epstopdf}
\usepackage{calc,pstricks, pgf, xcolor}
\usepackage{nicefrac}
\usepackage{enumerate}
\usepackage{dsfont}
\usepackage{bm}
\usepackage{color}
\usepackage{bbold}
\usepackage{setspace}
\usepackage{amsmath}

\newtheorem{OpenProblem}{Open Problem}

\newtheorem{theorem}{Theorem}

\newtheorem{remark}{Remark}

\newtheorem{proposition}{Proposition}
\newtheorem{lemma}{Lemma}

 \newtheorem{claim}{Claim}

\newcommand{\sign}{\mathop{\mathrm{sign}}}

\newcommand{\Ber}{\mathrm{Bernoulli}}
\newcommand{\Bin}{\mathrm{Binomial}}

\newcommand{\Unif}{\mathrm{Uniform}}
\newcommand{\KT}{\mathrm{KT}}

\newcommand{\Av}{\mathbf{A}}
\newcommand{\Bv}{\mathbf{B}}
\newcommand{\av}{\mathbf{a}}
\newcommand{\bv}{\mathbf{b}}

\def\Expt{\mathbb{E}}
\def\Prob{\mathbb{P}}
\def\b{\mathbf}

\def\dperp{\perp\!\!\!\perp}
\def\a{\alpha}
\def\b{\beta}
\def\d{\delta}
\def\e{\epsilon}

\def\t{\theta}

\def\T{\Theta}

\def\Gh{\widehat{G}}
\def\Real{{\mathbb R}}

\def\D{\Delta}
\def\Dt{\tilde{\Delta}}

\def\xt{{\widetilde x}} 
\def\Xt{{\widetilde X}} 

\def\gt{{\widetilde g}}
\def\qt{{\widetilde q}}
\def\Nt{{\widetilde N}}
\def\Kt{{\widetilde K}}

\def\Fc{{\mathcal F}}

\def\Gc{{\mathcal G}}

\def\Nc{{\mathcal N}}
\def\Pc{{\mathcal P}}

\def\Xc{{\mathcal X}}
\def\Yc{{\mathcal Y}}

\def\sphere{\mathbb{S}}
\def\Real{\mathbb{R}}
\def\Xv{{\bf X}}
\def\Th{\widehat{\Theta}}
\def\indic{\mathds{1}}

\def\ph{\widehat{p}}
\def\mix{\mathrm{mix}}
\def\xv{{\bf x}}

\title{Sequential prediction under log-loss with side information}

\author{
Alankrita Bhatt\\
University of California, San Diego\\
La Jolla, CA  92093, USA\\
\texttt{a2bhatt@eng.ucsd.edu}
\and Young-Han Kim\\
University of California, San Diego\\
La Jolla, CA  92093, USA\\
\texttt{yhk@ucsd.edu}
}

\date{}

\begin{document}

\maketitle

\begin{abstract}%
The problem of online prediction with sequential side information under logarithmic loss is studied, and general upper and lower bounds on the minimax regret incurred by the predictor is established. The upper bounds on the minimax regret are obtained by providing and analyzing a probability assignment inspired by mixture probability assignments in universal compression, and the lower bounds are obtained by way of a redundancy--capacity theorem. The tight characterization of the regret is provided in some special settings.
\end{abstract}

\section{Introduction}\label{sec:intro}
We consider a variant of the problem of sequential prediction under log-loss with side information. The particular variant under consideration was first studied in~\cite{Fogel--Feder17}. Let $X \in \Xc$ and $Y \in \{0,1\}$ denote two jointly distributed random variables. Let the marginal distribution of $X$ be denoted by $P_{X}(x)$. A hypothesis $f$ in the \emph{hypothesis class} $\Fc$ determines the conditional distribution $P_f(y|x)$, or equivalently, the conditional probability mass function (pmf) $p_f(y|x)$, for $y \in \{0,1\}$ and $x \in \Xc$. Each hypothesis is characterized by a tuple $f = (g, \t_0, \t_1)$ where
\begin{enumerate}
    \item $\t_0, \t_1 \in [0,1]$
    \item $g \in \Gc \subset \{\Xc \to \{0,1\}\}$.
\end{enumerate}
In other words, $g$ belongs to a class $\Gc$ of binary functions. We assume that $\Gc$ has finite VC dimension, denoted by VCdim$(\Gc)$.

Given a chosen hypothesis $f = (g,\t_0, \t_1)$ we then have 
\[
Y|\{X=x\} \sim \Ber(\t_{g(x)}).
\]
Thus, given the \emph{side information} $X$, the random variable $Y$ is distributed as either $\Ber(\t_0)$ or $\Ber(\t_1)$. Picking a hypothesis $f \in \Fc$, let $(X_i,Y_i)_{i=1}^n$ be drawn i.i.d.\@ from the joint distribution of $X$ and $Y$ characterized by the hypothesis $f$, so
\begin{align}\label{eq:XnYnLawh}
P(x^n,y^n) = \prod_{i=1}^n P_X(x_i)P_f(y_i|x_i).
\end{align}
The problem of sequential prediction under log-loss, also known as the sequential probability assignment problem, can be thought of as a game between the player and nature. First, nature picks a hypothesis $f \in \Fc$ unbeknownst to the player, and $X^n,Y^n$ are then generated according to the law~\eqref{eq:XnYnLawh}. At each time step $i \in [n]$,
$X_i$ is revealed to the player, who then assigns a probability mass function (pmf) $q(\cdot|X^i,Y^{i-1})$ to $Y_i$.  Next, $Y_i$ is revealed and the player incurs loss $-\log q(\cdot|X^i,Y^{i-1})$. Nature assigns the pmf $p_f(\cdot|X_i)$ at each time step $i$ and incurs loss $-\log p_f(Y_i|X_i)$. The goal of the game is to minimize the expected value of cumulative loss relative to nature (known as the regret), without knowledge of $f$. Importantly, we also wish to do this without knowing $P_X$ either.

To make this notion precise, define the regret incurred by the probability assignment $q$ when nature picked $f$ and the distribution of $X$ is $P_X$ as
\begin{align}\label{eq:minmaxprobFixHq}
R_{n,P_X}(q,f) := \Expt\left[\sum_{i=1}^n \log\frac{1}{q(Y_i|X^i,Y^{i-1})} - \sum_{i=1}^n\log\frac{1}{p_f(Y_i|X_i)}\right].
\end{align}
Then, the worst-case regret for the probability assignment $q$ is 
\begin{align}\label{eq:minmaxprobFixq}
    R_{n}(q) := \max_{P_X, f} R_{n,P_X}(q,f).
\end{align}
In this paper, we aim to calculate the min-max regret 
\begin{align}\label{eq:minmaxprob}
    R_n := \min_q  R_{n}(q).
\end{align}
and discover a probability assignment $q$ that is optimal or near-optimal in the sense of achieving $R_n(q)$ close to the optimal value~\eqref{eq:minmaxprob}.

The log-loss is of central importance in information theory as it connects two canonical problems in data science---compression and prediction; see the survey~\cite{Feder--Merhav98}. To motivate the use of the log-loss in the current problem, we view it as an extension of the problem of universal compression. Indeed, if there is no side information $X$ present, then the problem is equivalent to universal compression of an i.i.d. Bernoulli source which has been well studied~\cite{Rissanen83a, Rissanen83b, Rissanen84, Xie--Barron97, Xie--Barron00}. The minimax regret $R_n$ then is  significant operationally, representing the number of extra bits above the entropy one must pay as the price for compressing the source without knowing its distribution. Remarkably, one can show that $R_n = \frac{1}{2}\log n + o(\log n)$ in this setting. In a similar vein,~\cite{Shkel--Raginsky--Verdu18} studies a closely related problem where a compressed version of the sequence $Y^n$ is available as side information noncausally (i.e. not sequentially) and demonstrate its equivalence to lossy compression.

In the current setting, if the function $g$ is known, then simple extensions of the techniques developed to tackle the problem of universal compression of an i.i.d. Bernoulli source can be used to show that $R_n \le \log n + o(\log n)$, and we will elaborate on this important special case in detail in Section~\ref{subsec:onlyonefn}. The problem becomes nontrivial when the function $g$ is not known, and new techniques need to be developed to characterize $R_n$ in this case. 

In the standard study of classification in statistical learning theory, the loss function employed is the 0-1 loss or the indicator loss, and the notion of VC dimension plays a crucial role in characterizing the fundamental limits of binary classification~\cite{shalev2014understanding}. In particular, VCdim$(\Gc) < \infty$ implies the PAC-learnability of the hypothesis class $\Gc$. Viewing the current setting as a log-loss variant of the standard classification problem studied in statistical learning (which uses the indicator loss) motivates the choice of constraint VCdim$(\Gc) < \infty$. A variant of the current problem with indicator loss instead of log-loss was studied in~\cite{Lazaric--Munos12}. 
We have considered a specific class of conditional distributions to compete against ( recall that under hypothesis $f$ we have $p_f(Y=0|X=x) = \mathrm{Bern}(\t_{g(x)})$). As mentioned in the preceding paragraphs, our motivation stems from universal compression with side information, and to consider a log-loss variant of the standard binary classification problem. In both these cases, the choice of the considered class seems natural. However, in general, one could view this problem as an online conditional density estimation problem and correspondingly consider an arbitrary class $\Fc$ where any $f \in \Fc$ may characterize the conditional distribution $p_f(y|x)$ in a far more complex manner. It then makes sense to expect $R_n$ in this case to depend on a measure of complexity of $\Fc$ akin to the VC dimension. Indeed, in~\cite{RSTMartingaleLLN15} the authors develop a remarkable theory parallel to statistical learning theory when the data is non-i.i.d. They develop analogues of several combinatorial dimensions and the Rademacher complexity in the non-i.i.d. case. They then leverage this theory in~\cite{rstOLSeqComp15} to study the minmax regret in several online learning problems (with adversarial data). This approach is employed to study sequential prediction with the log-loss in~\cite{RSLogLoss15}  and~\cite{Bilodeau--Foster--Roy20}. However, it is important to note that the proofs in these works are nonconstructive—they proceed via using minmax duality and analyzing the dual game, which does not provide a strategy (i.e. a probability assignment) achieving the regret upper bound that is proven. Our method on the other hand involves construction of a sequential probability assignment. In the next subsection, we will mention and compare our results with the aforementioned two papers studying the log-loss. 
\subsection{Main Results}\label{subsec:mainresults}
Our first main result is a probability assignment that yields an upper bound on $R_n$.
\begin{theorem}\label{thm:genupperbd}
If $\Gc$ is such that VCdim$(\Gc) = d < \infty$, we have for an absolute constant $C \le 250$, for a probability assignment $q^*$ (which is specified in detail further on)
\begin{align}\label{eq:GenUpperBd}
R_n(q^*) \le  125 C\sqrt{dn}\log(2n) +  d(\log n)^2 + 2.
\end{align}
Moreover, for any $P_X,f, \d \in (0,1)$, with probability greater than $1-\d$,
\begin{align}\label{eq:GenUpperBdConc}
     \sum_{i=1}^n \log \frac{1}{q^*(Y_i|X^i,Y^{i-1})} - &\sum_{i=1}^n \log\frac{1}{p_{f}(Y_i|X_i)}\nonumber\\
     &\le  25 C\sqrt{dn}\log(2n) \left(C \sqrt{d} + \sqrt{2 \log \frac{2\log n}{\d}} \right) + d(\log n)^2 + 2
\end{align}
\end{theorem}

The proof is deferred to Section~\ref{sec:genupperbdpf}, where we construct and analyze the probability assignment $q^*$. In~\cite{Fogel--Feder17}, the authors established that $R_n = O(d \sqrt{n} \log n)$, and $R_n \le \left(2d+1+\log\frac{1}{\d}\right)\sqrt{n}\log n$ with probability $\ge 1-\d$. Our proof (and probability assignment) is different and achieves the same dependence on $n$, and a better dependence on $\d$ in the high-probability version of the result.

We also establish a lower bound on $R_n$. 
\begin{theorem}\label{thm:genlowerbd}
We have 
\[
R_n \ge d  + \log(n+1) - 2\sqrt{e} d^2 e^{-3n/100d} - \log(\pi e).
\]
\end{theorem}
The proof is deferred to Section~\ref{sec:lowerbd}. 

The non-constructive approaches of the papers~\cite{RSLogLoss15} and~\cite{Bilodeau--Foster--Roy20} mentioned earlier establish an $O(d \log n)$ upper bound for the $\Fc$ under consideration. In conjunction with Theorem~\ref{thm:genlowerbd} we see that the dependence of $R_n$ on $n$ is indeed $\Theta(\log n)$. This implies that the $q^*$ employed to prove Theorem~\ref{thm:genupperbd} is suboptimal and a better probability assignment could be constructed. 

\begin{OpenProblem}\label{op:ProbAssgnVC}
Construct a probability assignment $q$ for the VC hypothesis class that achieves $O(d \log n)$ regret. 
\end{OpenProblem}

As mentioned earlier, the problem of sequential probability assignment can be posed for any general (and possibly very complex) class $\Fc$ and viewed as an online conditional density estimation problem. 
\begin{OpenProblem}\label{op:ProbAssgnNonparametric}
Construct and analyze a probability assignment $q$ for the case when $\Fc$ is a general hypothesis class. 
\end{OpenProblem}
As a starting step towards Open Problem~\ref{op:ProbAssgnVC}, we considered a few special cases of the function class $\Gc$ in the hypothesis class  and provide a sequential probability assignment achieving $O(d \log n)$ upper bound. These upper bounds constitute our third main result. 

\subsection{Organization and Notation}
In Section~\ref{sec:MathPrelim} we provide basic notation and results that will be used in the proofs of our main results. In Section~\ref{sec:lognRegret}, we provide logarithmic upper bounds on $R_n$ for a few special cases. Section~\ref{sec:genupperbdpf} is devoted to the proof of Theorem~\ref{thm:genupperbd}, and Section~\ref{sec:lowerbd} is devoted to the proof of Theorem~\ref{thm:genlowerbd}. Finally, Section~\ref{sec:conclusion} concludes. All the proofs throughout the paper are relegated to the Appendix. 

\underline{Notation:} Throughout the paper, $\log(\cdot)$ refers to the logarithm to  base 2, and $\ln(\cdot)$ refers to logarithm to  base $e$. The Hamming distance between two binary vectors $x$ and $y$ is denoted by $d_{\mathrm{H}}(x,y)$.  The fact that two random variables $Z_1$ and $Z_2$ have the same distribution is denoted by $Z_1 \stackrel{(d)}{=} Z_2$. 

\section{Mathematical Preliminaries}\label{sec:MathPrelim}
This section introduces some basic notation and results that form the building blocks of the proofs of our main results. 

To prove any upper bound on $R_n$, it suffices to provide a probability assignment $q(Y_i|X^i,Y^{i-1})$ that achieves regret $R_n(q)$ that is less than the given upper bound. To this end, we will use a \emph{mixture probability assignment}
\begin{align}\label{eq:MixtureProbAssgn}
    q_{\mix}(y_i|x^i,y^{i-1}) := \frac{ \Expt_F[p_{F}\left(y^i|x^i\right)]}{ \Expt_F[p_{F}\left(y^{i-1}|x^{i-1}\right)]}
\end{align}
where $F = (\T_0,\T_1,G) \in \Fc$ is a random variable with some distribution over the hypothesis class $\Fc$. The usage of such a mixture probability assignment is inspired by previous work in universal prediction and universal compression, and we discuss this choice in further detail in Section~\ref{subsec:onlyonefn}. It can be verified that $q_{\mix}(y_i|x^i,y^{i-1})$ is indeed a probability assignment.

\begin{proposition}\label{prop:qmixisprobassgn}
For any $x^i,y^{i-1}$, we have $\sum_{y_i \in \Yc} q_{\mix}(y_i|x^i,y^{i-1}) = 1$.
\end{proposition}

For the probability assignment $q_{\mix}$ in~\eqref{eq:MixtureProbAssgn} we can establish the following.  
\begin{proposition}\label{prop:ChainRuleMixture}
We have
\begin{align}\label{eq:ChainRuleMixture}
   \sum_{i=1}^n \log\frac{1}{q_{\mix}(Y_i|X^i,Y^{i-1})} - \sum_{i=1}^n\log\frac{1}{p_f(Y_i|X_i)} = \log \frac{p_f(Y^n|X^n)}{\Expt[p_{F}(Y^n|X^n)]}.
\end{align}
\end{proposition}

The choice of the distribution of $F$ is important and greatly affects $R_n(q_{\mix})$. Almost all throughout this paper, for $F = (\T_0,\T_1,G)$, we will choose $\T_0$, $\T_1$ and $G$ to be mutually independent, with $\T_0, \T_1 \sim \mathrm{Beta}\left(\frac{1}{2}, \frac{1}{2}\right)$ each. The choice of the distribution of $G$ (which, recall, is over the class of functions $\Gc$) will be varied across different problems. The $\mathrm{Beta}\left(\frac{1}{2}, \frac{1}{2}\right)$ density is denoted by $w(\t) = \frac{1}{\pi \sqrt{\t(1-\t)}}$. The choice $w(\t)$ is elaborated upon in the next subsection. 


\subsection{When $|\Gc| = 1$}\label{subsec:onlyonefn}

In this subsection, we consider the rather simple case when the class of functions $|\Gc|$ contains only one function $g^*$ (or, equivalently, the function $g^*$ picked by nature is known). Thus, in this case, the hypothesis $f$ picked is of the form $(\t_0,\t_1,g^*)$. Considering $g^*$ to be the function for which $g^*(x) = 0$ $\forall$ $x \in \Xc$, as mentioned previously in the introduction, we recover the setting of universal compression over the class of binary i.i.d processes. In this case, the minmax regret $R_n$ in~\eqref{eq:minmaxprob} reduces to 
\begin{align}\label{eq:binaryiidregret}
    R_n = \min_q \max_{\t \in [0,1]} \Expt\left[\log \frac{p_{\t}(Y^n)}{q(Y^n)}\right] = \min_q \max_{\t \in [0,1]} D_{\mathrm{KL}}(p_{\t}(Y^n)||q(Y^n)) 
\end{align}
where $Y_i \sim \Ber(\t)$ i.i.d. and $p_{\t}(\cdot)$ is the probability law for this process. As mentioned in the introduction, it is well known that in this case 
\begin{align}\label{eq:binaryiidregvalue}
R_n = \frac{1}{2}\log n + o(\log n),
\end{align}
and that this is asymptotically achieved by an instance of the mixture probability assignment~\eqref{eq:MixtureProbAssgn} given by 
\[
q_{\KT}(y_i|y^{i-1}) = \frac{\int_{0}^1 p_{\t}(y^i) w(\t) d\t}{\int_{0}^1 p_{\t}(y^{i-1}) w(\t) d\t} = \frac{\Expt_{\T}[p_{\T}(y^i)]}{\Expt_{\T}[p_{\T}(y^{i-1})]}
\]
with $\T \sim \mathrm{Beta}(1/2,1/2)$. This probability assignment is known as the Krichevsky--Trofimov (KT) probability assignment~\cite{krichevsky1981performance} and motivates the use of the $\mathrm{Beta}$(1/2,1/2) prior for $\T_0$ and $\T_1$. For a sequence $y^n$, the sequential probability assignment $q_{\KT}(y_{i+1}|y^{i})$ turns out to be the so-called ``add-$1/2$" estimator which sets $q_{\KT}(0|y^{i}) = \frac{\sum_{t=1}^i \indic\{y_t = 0\} + 1/2}{i+1}$. Moreover, it can be shown that if $k = \sum_{i=1}^n y_i$,
\begin{align}
q_{\mathrm{KT}}(y^n) = \int_0^1 p_{\theta}(y^n) w(\t) d\theta = 
    \frac{1}{4^{n}}\frac{{n \choose k} {2n \choose n}}{{2n \choose 2k}}
\end{align}

When the range of $g^*$ includes both 0 and 1, a modification of the KT probability assignment can achieve regret $\log n + o(\log n)$. 

Consider the sequential probability assignment 
\[
q_{\KT}(0|x^i,y^{i-1}) =  \frac{\sum_{t=1}^{i-1} \indic\{y_t = 0, g^*(x_t) = g^*(x_i)\} + 1/2}{\sum_{t=1}^{i-1} \indic\{g^*(x_t) = g^*(x_i)\} + 1}
\]

Without the $x_i$, this can be seen to be the standard Krichevsky--Trofimov (KT) probability assignment for binary i.i.d. processes. With the side information $x_i$, $q_{\KT}$ is seen to be a ``block-wise" or ``symbol-wise" KT probability assignment. This can be seen to be a probability assignment of the form in~\eqref{eq:MixtureProbAssgn} with 
\[
q_{\KT}(y_i|x^i,y^{i-1}) = \frac{\int_{0}^1\int_{0}^1 p_{g^*,\t_0,\t_1}(y^i|x^i)w(\t_0)w(\t_1)d\t_0 d\t_1}{\int_{0}^1\int_{0}^1 p_{g^*,\t_0,\t_1}(y^{i-1}|x^{i-1})w(\t_0)w(\t_1)d\t_0 d\t_1}
\]
We can then bound the regret achieved by the probability assignment $q_{\KT}$.

\begin{lemma}\label{lem:BlkKTUpBd}
When the function class $\Gc$ is such that $|\Gc| = 1$, we have 
\begin{align}
    R_n(q_{\KT}) \le \log\left(\frac{n}{2} + 1\right) + \log\frac{\pi^2}{8}.
\end{align}
\end{lemma}


\begin{remark}[Laplace probability assignment]\label{rem:LaplaceProbAssgn}
Instead of using the $\mathrm{Beta}(1/2,1/2)$ prior, one can use the $\Unif[0,1]$ prior and choose the sequential probability assignment 
\[
q_{\mathrm{L}}(y_{i+1}|y^i) = \frac{\int_{0}^1 p_{\t}(y^{i+1}) d\t }{\int_{0}^1 p_{\t}(y^{i}) d\t},
\]
which yields the so-called Laplace or the add-1 probability assignment. It can be shown that for the problem~\eqref{eq:binaryiidregret}, $q_{\mathrm{L}}(\cdot)$ can achieve $R_n(q_{\mathrm{L}}) \le \log n + o(\log n)$.  Thus, the Laplace probability assignment achieves the optimal regret in order but with a slightly larger constant, a result that even holds for very rich expert classes~\cite{Yang--Barron99}.  It can be shown that if $k = \sum_{i=1}^n y_i$, we have $q_{\mathrm{L}}(y^n) = \int_0^1 p_{\theta}(y^n) d\theta = \frac{1}{(n+1) {n \choose k} }$. For mathematical convenience, we will use the Laplace probability assignment later in the paper, specifically in Sections~\ref{subsec:PxKnown} and~\ref{sec:genupperbdpf}.
\end{remark}

\subsection{When $|\Gc| < \infty$}\label{subsec:FiniteHypothesisClass}

When $|\Gc| < \infty$, we can use a probability assignment~\eqref{eq:MixtureProbAssgn} with $G$ distributed as $\Unif(\Gc)$. Then, for this choice of mixture, we have 
\begin{align}\label{eq:finiteSumMixture}
    \Expt_F\left[p_F\left(y^i|x^i\right)\right] = \frac{1}{|\Gc|}\sum_{g \in \Gc} \int_{0}^1 \int_{0}^1 p_{g, \t_0, \t_1}(y^i|x^i)w(\t_0)w(\t_1) d\t_0d\t_1
\end{align}
where $w(x) = \frac{1}{\sqrt{x(1-x)}}$ is the $\mathrm{Beta}(1/2,1/2)$ prior as before. 

We can then establish the following upper bound on the regret for the probability assignment $q_{\mix}$ characterized by the mixture~\eqref{eq:finiteSumMixture}.

\begin{lemma}\label{lem:finiteHypClass}
For the probability assignment $q_{\mix}$ with characterized by the mixture~\eqref{eq:finiteSumMixture}, we have 
\begin{align}
    R_n(q_{\mix}) \le \log |\Gc| + \log\left(\frac{n}{2}+1\right) + \log\frac{\pi^2}{8}.
\end{align}
\end{lemma}

\subsection{Side Information $X^n$ Available Noncausally}\label{subsec:noncausalsideinfo}
In this subsection we consider the special case when the side information $X^n$ is available \emph{noncausally} instead of sequentially. The results and intuition developed in this section will be used in proofs further ahead. 

When the side information $X^n$ is available noncausally, the probability assignment for $Y_i$ is of the form $q(Y_i|X^n,Y^{i-1})$ and the regret for a probability assignment $q$ can be seen to be
\begin{align}\label{eq:NoncausalRegret}
    R_{n,\mathrm{nc}}(q) = \max_{P_X,f} \Expt_{X^n,Y^n}\left[\sum_{i=1}^n \log \frac{1}{q(Y_i|X^n,Y^{i-1})} - \sum_{i=1}^n \log\frac{1}{p_f(Y_i|X_i)} \right].
\end{align}
Since the side information $X^n$ is available in advance, we can choose our mixture over the hypothesis class $\Fc$ to be dependent on $X^n$. As done so far, for $H = (\T_0,\T_1,G)$ we will choose $\T_0,\T_1$ and $G$ to be mutually independent with $\T_0, \T_1 \sim \mathrm{Beta}(1/2,1/2)$. We will now define a distribution over $\Gc$ that is dependent on the side information $X^n$. 

Given $X^n$, define the set $\Pc_n(X^n) = \{(g(X_1),\dotsc,g(X_n)), g \in \Gc \} \subseteq \{0,1\}^n$. For the remainder of this subsection, for brevity we will refer to $\Pc_n(X^n)$ by just $\Pc_n$. Enumerating the elements of $\Pc_n$ by $1,2,\dotsc, |\Pc_n|$, we now define the set
\begin{align}\label{eq:Ijdefn}
    I_j = \{g \in \Gc, (g(X_1),\dotsc,g(X_n)) = \Pc_n(j)\} 
\end{align}
where $\Pc_n(j)$ represents the $j$-th element in $\Pc_n$. Clearly, the sets $I_1,\dotsc,I_{|\Pc_n|}$ are nonempty and partition $\Gc$. So, $X^n$ can be thought of as partitioning $\Gc$ into sets where any two functions $g_1, g_2$ in the same partition have $g_1(X_j) = g_2(X_j) \forall j \in [n]$.

Pick an arbitrary $g_i \in I_i$ for $i = 1,\dotsc,|\Pc_n|$. Choosing $G \sim \Unif\left\{g_1,\dotsc,g_{|\Pc_n|}\right\}$, we have
\begin{align}\label{eq:noncausalMixture}
    \Expt_F\left[p_F\left(y^i|x^i\right)\right] = \frac{1}{|\Pc_n|}\sum_{i = 1}^{|\Pc_n|} \int_{0}^1 \int_{0}^1 p_{g_i, \t_0, \t_1}(y^i|x^i)w(\t_0)w(\t_1) d\t_0d\t_1.
\end{align}
\begin{lemma}\label{lem:noncausalsideInfo}
For the probability assignment $q_{\mix}$ characterized by the mixture~\eqref{eq:noncausalMixture}, we have 
\begin{align}\label{eq:regNoncausalVCdim}
    R_{n,\mathrm{nc}}(q_{\mix}) \le d\log(en/d) + \log\left(\frac{n}{2}+1\right) + \log \frac{\pi^2}{8}.
\end{align}
\end{lemma}
\subsection{When $P_X$ is Known}\label{subsec:PxKnown}

Consider the case when the distribution $P_X$ is known. In this case, the main idea is to choose the distribution of $G$ to be uniform over a \emph{finite} set of functions in $\Gc$ that form a fine-enough covering of $\Gc$. We make this idea precise next. First we will need the following Lemma.

\begin{lemma}[Lemma 13.6 of~\cite{Boucheron--Lugosi--Massart13}]\label{lem:VCClassCoveringNo}
For  $f,g \in \Gc$, define the metric $L^2(P_X)$ as 
\[
\|f-g\|_{L^2(P_X)} = \left( \Expt[f(X)-g(X)]^2\right)^{1/2}.
\]
Let $\Nc(\Gc,L^2(P_X),\e)$ denote the covering number of $\Gc$ in the metric $L^2(P_X)$. Then, we have 
\begin{align}\label{eq:VCClassCoveringNo}
    \Nc(\Gc,L^2(P_X),\e) \le \left(\frac{e^2}{\e}\right)^{2d}.
\end{align}
\end{lemma}

Consider the metric $d(f,g) = \Prob(g(X) \neq f(X))$ for $f,g \in \Gc$. Since 
\[
\|f-g\|_{L^2(P_X)} = \sqrt{\Prob(g(X) \neq f(X))},
\]
any $\sqrt{\e}$ covering of $\Gc$ in the $L^2(P_X)$ metric is a $\e$ covering of $\Gc$ in the metric $d$. Therefore,
\begin{align}\label{eq:VCClassCoveringNoDmetric}
    \Nc(\Gc,d,\e) \le \left(\frac{e^4}{\e}\right)^{d}.
\end{align}

We will now construct a mixture probability assignment of the form in~\eqref{eq:MixtureProbAssgn}. To do this, we must specify a distribution over the hypothesis class $\Fc$. Consider $g_1,g_2,\dotsc,g_{\lfloor (e^4n)^d \rfloor}$ that form a $1/n$ covering of $\Gc$ in the metric $d$. By~\eqref{eq:VCClassCoveringNoDmetric}, $\lfloor (e^4n)^d \rfloor$ such functions exist. Take $G, \T_0$ and $\T_1$ to be independent, with $\T_0, \T_1 \sim \Unif[0,1]$~\footnote{As mentioned in Remark~\ref{rem:LaplaceProbAssgn}, this corresponds to the Laplace probability assignment and we do this because it considerably simplifies the proof at just the cost of a slightly larger constant.} and $G \sim \Unif\{g_1,\dotsc,g_{\lfloor (e^4n)^d \rfloor}\}$. We then have 
\begin{align}\label{eq:VcclassCoverMix}
    \Expt_{F}[p_F(y^i|x^i)] = \frac{1}{\lfloor (e^4n)^d \rfloor} \sum_{i=1}^{ \lfloor (e^4n)^d \rfloor }\int_0^1 \int_0^1 p_{g_i,\t_0,\t_1} d\t_0 d\t_1
\end{align}
which we substitute into~\eqref{eq:MixtureProbAssgn} to construct $q_{\mix}$. We can then prove the following.

\begin{lemma}\label{lem:PxKnownReg}
For $q_{\mix}$ characterized by the mixture~\eqref{eq:VcclassCoverMix}, we have 
\begin{align}\label{eq:VCclassCovermixReg}
     \max_{f^*} R_{n,P_X}(q_{\mix},f^*) &\le (d+8)\log\left(e^4 n \right) + 6.
\end{align}
\end{lemma}
\section{Logarithmic upper bounds}\label{sec:lognRegret}
In this section, we consider some special instances of the function class $\Gc$ and distributions $P_X$ for which we can establish that the probability assignment $q_{\mix}$ in~\eqref{eq:MixtureProbAssgn} achieves $O(d \log n)$ regret for an appropriate choice of the mixture distribution (i.e. the distribution over the class $\Fc$). 

\subsection{Finite Function Class}\label{subsec:FiniteIsLogarithimic}
When $|\Gc| < \infty$, we have already shown in Lemma~\ref{lem:finiteHypClass} in Section~\ref{subsec:FiniteHypothesisClass} that the regret is logarithmic for any distribution $P_X$.  

\subsection{Function Class of Halfspaces}\label{subsec:Halfspaces}
In this subsection, we will consider the case when $\Gc$ is the class of \emph{halfspaces}, defined precisely as follows. Let $X \in \Xc = \sphere^{d-1}$. Recall that $\sphere^{d-1} = \{x \in \Real^d : \|x\|_2 = 1\}$. Define the function $g_{a}(x) : \Xc \to \{-1,1\}$ as $g_a(x) = \sign\left(a^T x\right)$. The class of functions $\mathrm{HS}_d := \{g_a, a \in \sphere^{d-1}\}$ is known as the class of $d-$dimensional (homogenous) halfspaces, and is known to have VCdim$(\Gc) = d$~\cite{shalev2014understanding}. Consider $X_1^n \sim \Unif(\sphere^{d-1})$ i.i.d. We will now evaluate the regret of $q_{\mix}$ in~\eqref{eq:MixtureProbAssgn}. 

As in the previous section, characterizing $q_{\mix}$ requires specifying a distribution over the hypothesis class $\Fc$, which in turn requires specifying a distribution over the function class $\mathrm{HS}_d$ (recall that $\T_0$ and $\T_1$ are chosen to be $\mathrm{Beta}(1/2,1/2)$ independently of each other and of $G$). We will choose $A \sim \Unif[\sphere^{d-1}]$. We then have  
\begin{align}
    \Expt_F[p_F(y^i|x^i)] = \Expt_{A}[\Expt_{\T_0,\T_1}[p_{A,\T_0,\T_1}(y^i|x^i)]]
\end{align}
Now, using the notation
\begin{align}\label{eq:mixoverThetasOnlyHS}
   q_{a,\mix}(y^i|x^i) := \Expt_{\T_0,\T_1}[p_{a,\T_0,\T_1}(y^i|x^i)] = \int_{0}^1\int_{0}^1 p_{a,\t_0,\t_1}(y^i|x^i)w(\t_0)w(\t_1) d\t_0 d\t_1
\end{align}
for an $a \in \sphere^{d-1}$, we see that 
\begin{align}\label{eq:MixtureProbHS}
    \Expt_F[p_F(y^i|x^i)] = \Expt_{A}[q_{A,\mix}(y^i|x^i)].
\end{align}
where $A \sim \Unif[\sphere^{d-1}]$ as mentioned previously. We can make the following assertion.

\begin{proposition}\label{prop:Halfspaceslogn}
If $P_X = \Unif[\sphere^{d-1}]$, then for the mixture probability assignment $q_{\mix}$ as defined in~\eqref{eq:MixtureProbAssgn}, with choice of mixture as in~\eqref{eq:MixtureProbHS}, we have 
\begin{align}
   \max_{f} R_{n,P_X}(q_{\mix},f) \le (2d +1) \log n + d \log(48d) + \log\frac{\pi^2}{8}. 
\end{align}
\end{proposition}

\subsection{Hypothesis Class of Axis-Aligned Rectangles}\label{subsec:AxisAlignedRects}
In this subsection, we will consider the case when $\Gc$ is the class of axis-aligned rectangles, defined precisely as follows. For\footnote{In this subsection, for clarity we will use boldface to denote vectors.} ${\bf a} := \{a_i\}_{i=1}^d$ and ${\bf b} := \{b_i\}_{i=1}^d$ that are such that $0 \le a_i \le b_i \le 1, i \in [d]$ define the function $g_{{\bf a},{\bf b}} : \Real^d \to \{0,1\}$ as $g_{{\bf a},{\bf b}}(\xv) = \prod_{i=1}^d \indic\{a_i \le x_i \le b_i\}$. Then the hypothesis class $\mathrm{RECT}_d := \{g_{{\bf a},{\bf b}}, {\bf a},{\bf b} \in [0,1]^d , a_i \le b_i\}$ is known as the class of axis aligned rectangles. It is well-known that VCdim(RECT$_d) = 2d$~\cite{shalev2014understanding}. Consider $\Xv_{1}^n \sim \Unif[0,1]^d$ iid. We will then evaluate the regret of the probability assignment $q_{\mix}$ in~\eqref{eq:MixtureProbAssgn}. 

As before, characterizing $q_{\mix}$ requires specifying a distribution over the hypothesis class $\Fc$, which in turn requires specifying a distribution over the function class $\mathrm{RECT}_d$ (recall that $\T_0$ and $\T_1$ are chosen to be $\mathrm{Beta}(1/2,1/2)$ independently of each other and of $G$). We will chose $(A_i,B_i) \sim \Unif\{(a,b) \in [0,1]\times[0,1], b \ge a\}$, and $(A_i,B_i) \perp \!\!\! \perp  (A_j,B_j)$ for $i \neq j$. Denoting $\Av := (A_1,\dotsc,A_d)$ and $\Bv := (B_1,\dotsc,B_d)$, for the aforementioned choice of distribution over $\Fc$, we have 
\begin{align}
    \Expt_F[p_F(y^i|\xv^i)] = \Expt_{\Av,\Bv}[\Expt_{\T_0,\T_1}[p_{\Av,\Bv,\T_0,\T_1}(y^i|\xv^i)]]
\end{align}
Now, using the notation
\begin{align}\label{eq:mixoverThetasOnly}
   q_{\av,\bv,\mix}(y^i|\xv^i) := \Expt_{\T_0,\T_1}[p_{\av,\bv,\T_0,\T_1}(y^i|\xv^i)] = \int_{0}^1\int_{0}^1 p_{\av,\bv,\t_0,\t_1}(y^i|\xv^i)w(\t_0)w(\t_1) d\t_0 d\t_1
\end{align}
we see that 
\begin{align}\label{eq:MixtureProbRects}
    \Expt_F[p_F(y^i|\xv^i)] = \Expt_{\Av,\Bv}[q_{\Av,\Bv,\mix}(y^i|\xv^i)].
\end{align}
We can then make the following assertion.

\begin{proposition}\label{prop:AxisAlignedRectslogn}
If $P_X = \Unif[0,1]^d$, then for  the probability assignment $q_{\mix}$ as defined in~\eqref{eq:MixtureProbAssgn}, with choice of mixture as in~\eqref{eq:MixtureProbRects}, we have 
\begin{align}
    \max_{f \in \Fc} R_{n,P_X}(q_{\mix},f) \le (2d+1)\log (n+1) + \log\frac{\pi^2}{8}. 
\end{align}
\end{proposition}
    
\begin{remark}
In Sections~\ref{subsec:Halfspaces} and~\ref{subsec:AxisAlignedRects}, we have fixed $P_X$ to be the uniform distribution. Considering the proofs, it appears to be a reasonable guess that the mixture probability assignment $q_{\mix}$ employed to prove the regret guarantees would work for other distributions $P_X$ that are sufficiently ``smooth". Thus, finding non-uniform $P_X$ for which the given $q_{\mix}$ achieves logarithmic regret is an intriguing question.
\end{remark}

\section{Proof of Theorem~\ref{thm:genupperbd} }\label{sec:genupperbdpf}
In this section, we prove Theorem~\ref{thm:genupperbd}. 
To motivate the main proof idea, recall the case discussed in Section~\ref{subsec:noncausalsideinfo} when noncausal side information is available. In that case, using the Sauer--Shelah lemma we argued that given $X^n$, the (possibly infinite) class of functions $\Gc$ could be effectively reduced to a class of at most $\left(\frac{en}{d}\right)^d$ functions, and using the mixture probability assignment that took a uniform mixture over these functions yielded an $O(d \log n)$ regret. This leads to us considering the following alternative to noncausal side information being available: what if \emph{another} sequence $\Xt^n \stackrel{(d)}{=} X^n$ is available noncausally? The sequence $\Xt^n$ also reduces the class $\Gc$ to at most $\left(\frac{en}{d}\right)^d$ functions (albeit not the same reduction as that of $\Gc$ by $X^n$). We establish in Section~\ref{subsec:auxseq} that a uniform mixture over the finite reduction of $\Gc$ induced by $\Xt^n$ achieves an $O(\sqrt{dn}\log n)$ regret. We then use this result in Section~\ref{subsec:epochmix} to establish a general $O(\sqrt{nd}\log n)$ regret when the side information $X^n$ is available sequentially. 

For clarity, throughout this section we will use $g(Z^n)$ to denote $(g(Z_1),\dotsc,g(Z_n)) \in \{0,1\}^n$. 

\subsection{Sequence $\Xt^n$ available noncausally}\label{subsec:auxseq}

Consider a sequence $\Xt^n \stackrel{(d)}{=} X^n, \Xt^n \dperp X^n$, with $X_i$ having distribution $P_X$ iid. In this subsection we consider the regret 
\begin{align}\label{eq:regretauxnoncausal}
\widetilde{R}_{n,P_X}(q,f) := \Expt_{X^n,Y^n,\Xt^n}\left[\sum_{i=1}^n \log\frac{1}{q(Y_i|X^i,Y^{i-1},\Xt^n)} - \sum_{i=1}^n\log\frac{1}{p_f(Y_i|X_i)}\right]
\end{align}
and in particular the worst-case regret attained by a probability assignment $q$ 
\begin{align}\label{eq:noncausalauxregminmax}
   \widetilde{R}_n(q) :=  \max_{f \in \Fc, P_X} \widetilde{R}_{n,P_X}(q,f).
\end{align}
Now, using the same notation as in Section~\ref{subsec:noncausalsideinfo}, let $\Pc_n(\Xt^n) = \{g(\Xt^n), g \in \Gc\} \subseteq \{0,1\}^n$ with $|\Pc_n(\Xt^n)| \le \left(\frac{en}{d}\right)^d$ by the Sauer--Shelah lemma. 
Pick $\gt_1,\gt_2,\dotsc,\gt_{|\Pc_n(\Xt^n)|} \in \Gc$ such that $\gt_j(\Xt^n) \in \Pc_n(\Xt^n)$, and $\gt_i(\Xt^n) \neq \gt_j(\Xt^n)$ if $i \neq j$.
Thus, for every $g \in \Gc$, there exists a $j \in \left[|\Pc_n(\Xt^n)|\right]$ such that $g(\Xt^n) = \gt_j(\Xt^n)$. 
Therefore, the class $\Gc$ has been effectively reduced to $|\Pc_n(\Xt^n)|$ functions by $\Xt^n$.
Consider now a mixture probability assignment, akin to~\eqref{eq:MixtureProbAssgn}, as 
\begin{align}\label{eq:mixtureauxprobassgn}
    \qt_{\mix}(y_i|x^i,y^{i-1},\xt^n) := \frac{\frac{1}{|\Pc(\xt^n)|}\sum_{j=1}^{|\Pc(\xt^n)|} \int_{0}^1 \int_{0}^1 p_{\gt_j,\t_0,\t_1}(y^i|x^i) d\t_0 d\t_1}{\frac{1}{|\Pc(\xt^n)|}\sum_{j=1}^{|\Pc(\xt^n)|} \int_{0}^1 \int_{0}^1 p_{\gt_j,\t_0,\t_1}(y^{i-1}|x^{i-1}) d\t_0 d\t_1}.
\end{align}
Note that this is indeed a mixture probability assignment in the sense of~\eqref{eq:MixtureProbAssgn}---$F = (G,\T_0,\T_1)$ has the distribution where $G \sim \Unif\{\gt_1,\dotsc,\gt_{|\Pc_n(\Xt^n)|}\}$, $\T_0, \T_1 \sim \Unif[0,1]$\footnote{The choice of taking a uniform prior for $\T_0$ and $\T_1$ instead of the Jeffreys prior is simply because using the uniform prior (which, recalling Remark~\ref{rem:LaplaceProbAssgn}, corresponds to the Laplace probability assignment) makes some calculations far simpler in the proof at just the cost of a worse constant factor in the regret.} and $G,\T_0,$ and $\T_1$ are mutually independent. We can now state the following.
\begin{lemma}\label{lem:regauxnoncausal}
For $\qt_{\mix}$ defined in~\eqref{eq:mixtureauxprobassgn}, we have for an absolute constant $C \le 250$,
\begin{align}
    \widetilde{R}_n(\qt_{\mix}) \le d\log (en/d) + 16C\sqrt{nd}\log(6n+2)
\end{align}
and moreover, for any $P_X$ and $h$ we have 
\begin{align}
   \sum_{i=1}^n \log\frac{1}{\qt_{\mix}(Y_i|X^i,Y^{i-1},\Xt^n)} - &\sum_{i=1}^n\log\frac{1}{p_f(Y_i|X_i)} 
    \nonumber\\
    &\le d\log (en/d) + 16\sqrt{n}\log(6n+2)\left(C\sqrt{d} + \sqrt{2 \log \frac{2}{\d}}\right).
\end{align}
\end{lemma}

\begin{remark}[Empirical covering]
The probability assignment $\qt_{\mix}$ can also be motivated by considering the scenario in Section~\ref{subsec:PxKnown} where $P_X$ is known. Recall that there, we took a uniform mixture over a $1/n$-covering of $\Gc$ in the metric $d$ with $d(g_1,g_2) = \Prob(g_1(X) \neq g_2(X))$. If we have $\Xt^n$, as an alternative to a mixture over a covering in the metric $d$, we can take an empirical $1/n$ covering of $\Gc$, i.e. a covering in the metric $\Dt_n(g_1,g_2) = \frac{1}{n}d_H(g_1(\Xt^n),g_2(\Xt^n))$. Indeed, the functions $\gt_1,\gt_2,\dotsc,\gt_{|\Pc_n(\Xt^n)|}$ form not just a $1/n$ covering but a 0-covering  of $\Gc$ in the metric $\Dt_n$.
\end{remark}

\subsection{Epoch-based mixture probability} \label{subsec:epochmix}
In this subsection, we use Lemma~\ref{lem:regauxnoncausal} to construct a general probability assignment when side information is available sequentially. In this scenario, we do not have access to another sequence $\Xt^n$. However, at time step $i+1$, we have access to the past sequence $X^{i}$ which could be used, as done in~\cite{Lazaric--Munos12}, in lieu of  $\Xt^i$. We now precisely define and analyze this probability assignment. 

For simplicity, assume that $n = 2^k$ for some integer $k$. The analysis is easily extended to any arbitrary $n$. We will split the $n$ time steps into $\log n$ ``epochs". Starting from $j = 1$, define the $j-$the epoch to consist of the time steps $2^{j-1} + 1 \le i \le 2^j$. So, the first epoch consists of $X_2$, the second epoch consists of $X_3^4$, the third epoch consists of $X_5^8$ and so on. Consider the the following probability assignment $q^*$. 
\begin{enumerate}
    \item $q^*(Y_1|X_1) = 1/2$
    \item For $i \ge 2$, if $2^{j-1} + 1 \le i \le 2^j$, i.e. if the time step $i$ falls within the $j-$th epoch, then
    \begin{align}
        q^*(Y_i|X^i,Y^{i-1}) = \frac{q_{\mix,j}(Y_{2^{j-1}+1}^i|X_{2^{j-1}+1}^i)}{q_{\mix,j}(Y_{2^{j-1}+1}^{i-1}|X_{2^{j-1}+1}^{i-1})}
    \end{align}
    where 
    \begin{align}
        q_{\mix,j}((Y_{2^{j-1}+1}^i&|X_{2^{j-1}+1}^i))  \nonumber\\
        &:= \frac{1}{|\Pc(X^{2^{j-1}})|}\sum_{k=1}^{|\Pc(X^{2^{j-1}})|} \int_0^1 \int_0^1 p_{\t_0,\t_1,g_k}(Y_{2^{j-1}+1}^i|X_{2^{j-1}+1}^i) d\t_0 d\t_1
    \end{align}
    is the finite mixture over the $|\Pc(X^{2^{j-1}})|$ partition of $\Gc$ induced by $X^{2^{j-1}}$. This is the same probability assignment as in~\eqref{eq:mixtureauxprobassgn}.
\end{enumerate}
Using Lemma~\ref{lem:regauxnoncausal} and an epoch-wise analysis of $q^*$ we can establish Theorem~\ref{thm:genupperbd}.
\section{Proof of Theorem~\ref{thm:genlowerbd}}\label{sec:lowerbd}

In this section we prove Theorem~\ref{thm:genlowerbd}. A key component of the proof is the redundancy-capacity theorem~\cite{merhav1995strong}.  

First, note that the class of probability assignments that utilize the side information $X^n$ causally is a subset of the set of probability assignments that utilize the side information $X^n$ noncausally. This implies
\begin{align}
    R_n &= \min_q \max_{P_X, f} \Expt_{X^n,Y^n}\left[\sum_{i=1}^n \log\frac{1}{q(Y_i|X^i,Y^{i-1})} - \sum_{i=1}^n\log\frac{1}{p_f(Y_i|X_i)}\right] \nonumber \\
    &\ge \min_q \max_{P_X, f} \Expt_{X^n,Y^n}\left[\sum_{i=1}^n \log\frac{1}{q(Y_i|X^n,Y^{i-1})} - \sum_{i=1}^n\log\frac{1}{p_f(Y_i|X_i)}\right] \label{eq:causalsubsetnoncausal}
 \end{align}
 and therefore 
 \begin{align}
    R_n &\ge \min_q \max_{P_X, f} \Expt_{X^n,Y^n}\left[\log\frac{p_f(Y^n|X^n)}{q(Y^n|X^n)}\right] \nonumber\\
    &= \min_q \max_{P_X, P_F} \Expt_{F,X^n,Y^n}\left[\log\frac{p_F(Y^n|X^n)}{q(Y^n|X^n)}\right] \label{eq:innermaxmixture} \\
    &\ge \max_{P_X,P_F} \min_q \Expt_{F,X^n,Y^n}\left[\log\frac{p_F(Y^n|X^n)}{q(Y^n|X^n)}\right] \label{eq:minimaxIneq} 
\end{align}
where $P_F$ denotes a distribution over $\Fc$ in~\eqref{eq:innermaxmixture}, and~\eqref{eq:minimaxIneq} follows since $\min \max (\cdot) \ge \max \min (\cdot)$. By a conditional variant of the redundancy-capacity theorem we have for a fixed $P_X$ and $P_F$ (recall that $F = (\T_0, \T_1, G)$) 
\begin{align}\label{eq:CondRedCapThmOne}
\min_q \Expt_{F,X^n,Y^n}\left[\log\frac{p_F(Y^n|X^n)}{q(Y^n|X^n)}\right] = I(F;Y^n|X^n)
\end{align}
and so
\begin{align}\label{eq:ChoosePxPh}
    R_n &\ge \max_{P_X,P_F} I(F;Y^n|X^n)
\end{align}
where recall $F = (\T_0, \T_1, G)$.
\begin{remark}
The result in~\eqref{eq:ChoosePxPh} holds for any class of conditional distributions $\Fc$, not just the VC class under consideration.
\end{remark}
We will first provide a lower bound on $R_n$ when $|\Xc| < \infty$ which we will then use to achieve a general lower bound on $R_n$.  
\begin{lemma}\label{lem:MinimaxLBRegretFinite}
If $|\Xc| = m < \infty$ and $\Gc = 2^{[m]}$ so that $|\Gc| = 2^m$, we have 
\begin{align}
    R_{n} 
    &\ge  m + \log(n+1) - \log(\pi e) - 2\sqrt{e}m^2e^{-3n/100m}.
\end{align}
\end{lemma}
Lemma~\ref{lem:MinimaxLBRegretFinite} is proved by choosing a particular $P_X$ and $P_F$ and analyzing the right hand side of~\eqref{eq:ChoosePxPh}. 
\begin{remark}[Tightness for finite $\Xc$]
Combining Lemma~\ref{lem:MinimaxLBRegretFinite} and Lemma~\ref{lem:finiteHypClass} with $|\Gc| = 2^m$, we see that for $\Xc = m, \Gc = 2^{[m]}$, we can obtain a tight characterization of the regret $R_n$ on $n$ and $m$.
\end{remark}

Consider now the case when $\Xc$ is possibly infinite. Since VCdim$(\Gc) = d, \exists x_1,\dotsc,x_d \in \Xc$ such that $|\{(g(x_1),\dotsc,g(x_d)), g \in \Gc\}| = 2^d$. Theorem~\ref{thm:genlowerbd} then follows as a corollary to Lemma~\ref{lem:MinimaxLBRegretFinite} by substituting $m = d$ and choosing the distributions of $P_X, P_F$ as in the proof of Lemma~\ref{lem:MinimaxLBRegretFinite}.
\section{Discussion}\label{sec:conclusion}
We considered the problem of sequential prediction under log-loss with side information. This can be considered as an extension of the well-studied information-theoretic problem of universal compression of an i.i.d. binary source, and the regret incurred can be characterized via the value of a minmax game. We provided upper bounds on the regret via construction of a probability assignment, and lower bounds by the redundancy-capacity theorem. There are several open directions. Previous results established an $O(d \log n)$ upper bound via minmax duality. Even though our upper and lower bounds are off by a $\sqrt{n}$ factor, we suspect that a variant of the mixture probability assignment from information theory can achieve the optimal $O(d \log n)$ upper bound. We provided some special cases and a probability assignment where $O(d\log n)$ redundancy is achieved to provide evidence for this. It would also be interesting to answer the weaker question of whether the current upper bound on $R_n$ can be improved upon (constructively) under certain further restrictions on the class of functions $\Gc$. Moreover, even though the lower bound cannot be improved in order, it may be possible to get a better dependence on $d$. Finally, we have not considered complexity concerns for actual algorithmic implementation. Computing the coverings may be probihitively expensive in several cases, so finding efficient algorithms for sequential probability assignment is yet another avenue to be explored. All these directions are promising for further study. 


\appendix
\section{Skipped Proofs from Section~\ref{sec:MathPrelim}}

\subsection{Proof of Proposition~\ref{prop:qmixisprobassgn}}
\begin{align*}
  \sum_{y_i \in \Yc}  q_{\mix}(y_i|x^i,y^{i-1})  &=  \frac{\sum_{y_i \in \Yc} \Expt[p_{F}(y^i|x^i)]}{ \Expt[p_{F}(y^{i-1}|x^{i-1}))} \\
  &= \frac{\Expt[ \sum_{y_i \in \Yc}  p_{F}(y^i|x^i)]}{ \Expt[p_{F}(y^{i-1}|x^{i-1}))} \\
  &= \frac{\Expt[ p_{F}(y^{i-1}|x^{i-1})\sum_{y_i \in \Yc} p_{F}(y_i|x_i)]}{ \Expt[p_{F}(y^{i-1}|x^{i-1}))} \\
  &= \frac{\Expt[p_{F}(y^{i-1}|x^{i-1})]}{ \Expt[p_{F}(y^{i-1}|x^{i-1}))} = 1
\end{align*}
and so $q_{\mix}(y_i|x^i,y^{i-1})$ is a valid probability assignment. 

\subsection{Proof of Lemma~\ref{lem:BlkKTUpBd}}
Since the function $g^*$ is known and the range of $g^*$ is only 0 and 1, we can assume without loss of generality that the side information is binary, i.e. $\Xc = \{0,1\}$. Now define  
\begin{align}
     n_l &:= \sum_{i=1}^n \indic\{x_i = l\}, l \in \{0,1\} \label{eq:nlDef} \\
   k_l &:= \sum_{i=1}^n \indic\{y_i = 1, x_i = l\}, l \in \{0,1\}. \label{eq:klDef} 
\end{align}
Note that 
\begin{align}
  \log \frac{p_{\t_0,\t_1,g^*}(Y^n|X^n)}{\int_{0}^1\int_{0}^1 p_{\t_0,\t_1,g^*}(y^n|x^n)w(\t_0)w(\t_1)d\t_0} =  \sum_{i=1}^n \log\frac{1}{q_{\KT}(Y_i|X^i,Y^{i-1})} - \sum_{i=1}^n\log\frac{1}{p_f(Y_i|X_i)} 
\end{align}
and
\begin{align}
    \sum_{i=1}^n \log\frac{1}{q_{\KT}(Y_i|X^i,Y^{i-1})} &- \sum_{i=1}^n\log\frac{1}{p_f(Y_i|X_i)} \nonumber\\
    & = \sum_{l=0}^{1} \sum_{i: X_i = l}\left[ \log\frac{1}{q_{\KT}(Y_i|X^i,Y^{i-1})} - \log\frac{1}{p_f(Y_i|X_i)}\right] \nonumber\\
    &=\sum_{l=0}^{1} \left[\log \frac{\prod_{i: X_i = l}p_f(Y_i|X_i)}{\prod_{i: X_i = l}q_{\KT}(Y_i|X^i,Y^{i-1})} \right] \label{eq:blkWiseRegret}.
\end{align}
Now, if $n_l = 0$, we have 
\begin{align}\label{eq:nlZero}
    \log \frac{\prod_{i: X_i = l}p_f(Y_i|X_i)}{\prod_{i: X_i = l}q_{\KT}(Y_i|X^i,Y^{i-1})} = 0
\end{align}
And if $n_l \ge 1$, we have 
\begin{align}
    \log \frac{\prod_{i: X_i = l}p_f(Y_i|X_i)}{\prod_{i: X_i = l}q_{\KT}(Y_i|X^i,Y^{i-1})} = \log\frac{\t_{l}^{k_l}(1-\t_{l})^{n_l - k_l}}{\frac{1}{4^{n_l}}\frac{{n_l \choose k_l} {2n_l \choose n_l}}{{2n_l \choose 2k_l}}} \label{eq:KTcumulativeProb}
\end{align}
where~\eqref{eq:KTcumulativeProb} follows from properties of the KT sequential probability assignment. Moreover, we have 
\begin{align}\label{eq:maxLikelihood}
    \t_{l}^{k_l}(1-\t_{l})^{n_l - k_l} \le \left(\frac{k_l}{n_l}\right)^{k_l} \left(1 - \frac{k_l}{n_l}\right)^{n_l - k_l} = 2^{-n_l h\left(\frac{k_l}{n_l}\right)}
\end{align}
which can be established by noting that the binary KL divergence 
\[
d\left(\frac{k_l}{n_l}||\t_{l}\right) = \frac{1}{n_l}\log\frac{\left(\frac{k_l}{n_l}\right)^{k_l} \left(1 - \frac{k_l}{n_l}\right)^{n_l - k_l}}{\t_{g(l)}^{k_l}(1-\t_{g(l)})^{n_l - k_l}} \ge 0,
\] and furthermore, using a Sterling approximation we can establish 
\begin{align}\label{eq:KTProbBd}
    \frac{1}{4^{n_l}}\frac{{n_l \choose k_l} {2n_l \choose n_l}}{{2n_l \choose 2k_l}} \le \sqrt{\frac{8}{\pi^2}} \frac{2^{-n_l  h\left(\frac{k_l}{n_l}\right) } }{\sqrt{n_l}}.
\end{align}
Plugging~\eqref{eq:maxLikelihood} and~\eqref{eq:KTProbBd} into~\eqref{eq:KTcumulativeProb} yields 
\begin{align}\label{eq:NlnonzeroUB}
    \log \frac{\prod_{i: X_i = l}p_f(Y_i|X_i)}{\prod_{i: X_i = l}q_{\KT}(Y_i|X^i,Y^{i-1})} \le \log\sqrt{\frac{\pi^2}{8}}\frac{ 2^{-n_l h\left(\frac{k_l}{n_l}\right)} \sqrt{n_l}}{  2^{-n_l h\left(\frac{k_l}{n_l}\right)} } = \frac{1}{2} \log n_l + \frac{1}{2}\log \frac{\pi^2}{8} 
\end{align}
when $n_l \ge 1$. Combining~\eqref{eq:nlZero} and~\eqref{eq:NlnonzeroUB} we can establish 
\begin{align}\label{eq:CombinednlUB}
     \log \frac{\prod_{i: X_i = l}p_f(Y_i|X_i)}{\prod_{i: X_i = l}q_{\KT}(Y_i|X^i,Y^{i-1})}  &\le \frac{1}{2}\log(n_l + 1) + \frac{1}{2}\log\frac{\pi^2}{8}
\end{align}
for all $n_l \ge 0$. Plugging the upper bound~\eqref{eq:CombinednlUB} into~\eqref{eq:blkWiseRegret} yields
\begin{align}
    \sum_{l=0}^1 \left[\log \frac{\prod_{i: X_i = l}p_f(Y_i|X_i)}{\prod_{i: X_i = l}q_{\KT}(Y_i|X^i,Y^{i-1})} \right] &\le \sum_{l=0}^{1} \frac{1}{2}\log(n_l + 1) + \frac{1}{2}\log\frac{\pi^2}{8} \nonumber \\
    &= \frac{1}{2}\log \prod_{l=0}^{1}(n_l + 1) + \log\frac{\pi^2}{8} \nonumber \\
    &\le \frac{1}{2}\log\left(\frac{n}{2}+1\right)^2 +  \log\frac{\pi^2}{8} \label{eq:AMGMineq}\\
    &= \log\left(\frac{n}{2}+1\right) +  \log\frac{\pi^2}{8} 
\end{align}
where the inequality~\eqref{eq:AMGMineq} follows by noting that $\sum_{l=0}^{m-1}(n_l + 1) = n + m$ and the using the AM-GM inequality. We have now established 
\[
\sum_{i=1}^n \log\frac{1}{q_{\KT}(Y_i|X^i,Y^{i-1})} - \sum_{i=1}^n\log\frac{1}{p_f(Y_i|X_i)}  \le \log\left(\frac{n}{2}+1\right) + \log\frac{\pi^2}{8} 
\]
and monotonicity of expectation followed by taking supremum over $\t_0,\t_1$ then yields the result. 
\subsection{Proof of Lemma~\ref{lem:finiteHypClass}}

By Proposition~\ref{prop:ChainRuleMixture}, we see that for a fixed $f^* = (g^*, \t_0^*, \t_1^*)$ and $P_X$, the regret achieved by the probability assignment $q_{\mix}$ characterized by~\eqref{eq:finiteSumMixture} is
\[
R_{n,P_X}(q_{\mix},f^*) =\Expt_{X^n,Y^n}\left[ \log\frac{p_{f^*}(Y^n|X^n)}{\Expt_F\left[p_F(Y^n|X^n)\right]}\right] 
\]
and we have 
\begin{align}
  \Expt_{X^n,Y^n}&\left[ \log\frac{p_{f^*}(Y^n|X^n)}{\Expt_F\left[p_F(Y^n|X^n)\right]}\right] \nonumber\\
   &\qquad=  \Expt\left[\log\frac{p_{f^*}(Y^n|X^n)}{\frac{1}{|\Gc|}\sum_{g \in \Gc} \int_{0}^1 \int_{0}^1 p_{g, \t_0, \t_1}(Y^n|X^n)w(\t_0)w(\t_1) d\t_0d\t_1}\right] \nonumber \\
    &\qquad= \log |\Gc| + \Expt\left[\log \frac{p_{f^*}(Y^n|X^n)}{\sum_{g \in \Gc} \int_{0}^1 \int_{0}^1 p_{g, \t_0, \t_1}(Y^n|X^n)w(\t_0)w(\t_1) d\t_0d\t_1}\right] \nonumber \\ 
    &\qquad\le \log |\Gc| + \Expt\left[\log \frac{p_{f^*}(Y^n|X^n)}{\int_{0}^1 \int_{0}^1 p_{g^*, \t_0, \t_1}(Y^n|X^n)w(\t_0)w(\t_1) d\t_0d\t_1}\right] \label{eq:ProbsGeqZero} \\
    &\qquad\le  \log |\Gc| + \log\left(\frac{n}{2}+1\right) + \log\frac{\pi^2}{8} \label{eq:fixedGmixthetas}
\end{align}
where~\eqref{eq:ProbsGeqZero} follows since each of the summands in the denominator of the second term are nonnegative, and~\eqref{eq:fixedGmixthetas} is a consequence of Lemma~\ref{lem:BlkKTUpBd}. 


\subsection{Proof of Lemma~\ref{lem:noncausalsideInfo}}

Following the proof of Lemma~\ref{lem:finiteHypClass} up to~\eqref{eq:ProbsGeqZero}, for any fixed $f^* = (g^*,\t_0^*,\t_1^*)$ the probability assignment $q_{\mix}$ characterized by the mixture~\eqref{eq:noncausalMixture} has 
\begin{align}
\Expt_{X^n,Y^n}&\left[ \sum_{i=1}^n \log \frac{1}{q_{\mix}(Y_i|X^n,Y^{i-1})} - \sum_{i=1}^n \log\frac{1}{p_{f^*}(Y_i|X_i)} \right] \nonumber\\ &\qquad\qquad\le \Expt\left[|\Pc_n|\right] + \Expt\left[\log \frac{p_{f^*}(Y^n|X^n)}{\int_{0}^1 \int_{0}^1 p_{g_j, \t_0, \t_1}(Y^n|X^n)w(\t_0)w(\t_1) d\t_0d\t_1}\right] \label{eq:useSauerLemma}
\end{align}

Where $j \in \left[|\Pc_n|\right]$ is such that 
\[
(g^*(X_1),\dotsc,g^*(X_n)) = (g_j(X_1),\dotsc,g_j(X_n)).
\]
If VCdim$(\Gc) = d < \infty$, we can control $|\Pc_n|$ using the following standard result~\cite[Chapter 8]{vershynin2018high}.
\begin{lemma}[Sauer--Shelah]\label{lem:SauerShelah}
If VCdim($\Gc) = d < \infty$, then $|\Pc_n| \le \left(\frac{en}{d}\right)^d$.
\end{lemma}
Finally, using Lemma~\ref{lem:BlkKTUpBd} and Lemma~\ref{lem:SauerShelah} in~\eqref{eq:useSauerLemma} yields
\begin{align}\label{eq:regNoncausalPartitionSize}
    R_{n,\mathrm{nc}}(q_{\mix}) \le d \log (en/d)  + \log\left(\frac{n}{2}+1\right) + \log\frac{\pi^2}{8}.
\end{align}

\subsection{Proof of Lemma~\ref{lem:PxKnownReg}}

By Proposition~\ref{prop:ChainRuleMixture}, we have for a fixed $f^* = (g^*, \t_0^*, \t_1^*)$ 
\begin{align}
    R_{n,P_X}(q_{\mix},f^*) &= \Expt\left[\frac{p_{f^*}(Y^n|X^n)}{ \Expt_{H}[p_F(Y^n|X^n)]}\right] \nonumber \\
    &= \Expt \left[\log \frac{p_{f^*}(Y^n|X^n) }{\frac{1}{\lfloor (e^4n)^d \rfloor} \sum_{i=1}^{\lfloor (e^4n)^d\rfloor }\int_0^1 \int_0^1 p_{g_i,\t_0,\t_1}(Y^n|X^n) d\t_0 d\t_1}\right] \nonumber \\
    &\le d \log  (e^4n)  + \Expt\left[\log \frac{ p_{f^*}(Y^n|X^n)}{\sum_{i=1}^{\lfloor (e^4n)^d\rfloor }\int_0^1 \int_0^1 p_{g_i,\t_0,\t_1}(Y^n|X^n) d\t_0 d\t_1}\right] \label{eq:SubVcCovering} 
\end{align}
Let $\gt \in \{g_1,g_2,\dotsc,g_{\lfloor (e^4n)^d \rfloor}\}$ be such that $\Prob(\gt(X) \neq g^*(X)) = d(\gt,g^*) \le 1/n$. Such a $\gt$ exists since $g_1,\dotsc,g_{\lfloor (e^4n)^d\rfloor}$ form a $1/n$ covering of $\Gc$ in the metric $d$. We then have from~\eqref{eq:SubVcCovering} 
\begin{align}
  R_{n,P_X}(q_{\mix},f^*) &\le  d \log (e^4n) + \Expt\left[\log \frac{ p_{f^*}(Y^n|X^n)}{\int_0^1 \int_0^1 p_{\gt,\t_0,\t_1}(Y^n|X^n) d\t_0 d\t_1}\right] \label{eq:SubVcCoveringTwo}.
\end{align}
Now, defining
\begin{align}
     N_j &:= \sum_{i=1}^n \indic\{g^*(X_i) = j\}, j \in \{0,1\}  \nonumber \\
   K_j &:= \sum_{i=1}^n \indic\{g^*(X_i) = j, Y_i = 1\}, j \in \{0,1\} \nonumber \\
    \Nt_j &:= \sum_{i=1}^n \indic\{\gt(X_i) = j\}, j \in \{0,1\}  \nonumber \\
   \Kt_j &:= \sum_{i=1}^n \indic\{\gt(X_i) = j, Y_i = 1\}, j \in \{0,1\} \nonumber
\end{align} 
we have 
\begin{align} \label{eq:IIDconditionalsProb}
p_{f^*}(Y^n|X^n) = p_{g^*,\t_0^*,\t_1^*}(Y^n|X^n) = \t_0^{*K_0}(1-\t_0^*)^{N_0-K_0}\t_1^{*K_1}(1-\t_1^*)^{N_1-K_1}
\end{align}
and 
\begin{align}
    \int_{0}^1 \int_{0}^1 p_{\gt,\t_0,\t_1}(Y^n|X^n) d\t_0 d\t_1 &=  \int_{0}^1 \int_{0}^1 \t_0^{\Kt_0}(1-\t_0)^{\Nt_0-\Kt_0}\t_1^{\Kt_1}(1-\t_1)^{\Nt_1-\Kt_1}  d\t_0 d\t_1 \nonumber \\
    &= \int_{0}^1  \t_0^{\Kt_0}(1-\t_0)^{\Nt_0-\Kt_0} d\t_0 \int_0^1 \t_1^{\Kt_1}(1-\t_1)^{\Nt_1-\Kt_1}  d\t_1 \nonumber \\
    &= \frac{1}{(\Nt_0+1){\Nt_0 \choose \Kt_0}(\Nt_1+1){\Nt_1 \choose \Kt_1}} \label{eq:LaplaceMixtureIntegral}
\end{align}
where~\eqref{eq:LaplaceMixtureIntegral} follows from properties of the Laplace probability assignment. Now, from~\eqref{eq:IIDconditionalsProb} and~\eqref{eq:LaplaceMixtureIntegral}, we have 
\begin{align}
   &\frac{p_{f^*}(Y^n|X^n)}{\int_0^1 \int_0^1 p_{\gt,\t_0,\t_1}(Y^n|X^n) d\t_0d\t_1} \nonumber\\
   &\qquad= (\Nt_0+1)(\Nt_1+1){\Nt_0 \choose \Kt_0}\t_0^{*K_0}(1-\t_0^*)^{N_0-K_0}{\Nt_1 \choose \Kt_1}\t_1^{*K_1}(1-\t_1^*)^{N_1-K_1} \nonumber \\
    &\qquad\le (n+1)^2 {\Nt_0 \choose \Kt_0}\t_0^{*K_0}(1-\t_0^*)^{N_0-K_0}{\Nt_1 \choose \Kt_1}\t_1^{*K_1}(1-\t_1^*)^{N_1-K_1}  \label{eq:nOlessn} \\
    &\qquad\le (n+1)^2 \frac{{\Nt_0 \choose \Kt_0}}{{N_0 \choose K_0}} \frac{{\Nt_1 \choose \Kt_1}}{{N_1 \choose K_1}} \label{eq:binomialproblessone}
\end{align}
where~\eqref{eq:nOlessn} follows because $\Nt_0,\Nt_1 \le n$, and~\eqref{eq:binomialproblessone} follows since ${n \choose k}x^k(1-x)^{n-k} \le 1$ for any $x \in [0,1]$.  Substituting~\eqref{eq:binomialproblessone} into~\eqref{eq:SubVcCoveringTwo} yields  
\begin{align}
     R_{n,P_X}(q_{\mix},f^*) \le d \log (e^4n) + \Expt_{X^n,Y^n}\left[\log \frac{{\Nt_0 \choose \Kt_0}}{{N_0 \choose K_0}}\right] +  \Expt_{X^n,Y^n}\left[\log \frac{{\Nt_1 \choose \Kt_1}}{{N_1 \choose K_1}}\right] \label{eq:ratiosofnchoosek}
\end{align}
Recall that $d_H(\cdot,\cdot)$ denotes the Hamming distance. We can then easily verify the following proposition.
\begin{proposition}\label{prop:HammDistNK}
We have 
\[
|\Nt_j - N_j| \le d_H(g^*(X^n),\gt(X^n)), j \in \{0,1\} \text{ and }|\Kt_j - K_j| \le d_H(g^*(X^n),\gt(X^n)), j \in \{0,1\}.
\]
\end{proposition}
We now wish to use Proposition~\ref{prop:HammDistNK}, to obtain a bound on $\log \frac{{\Nt_0 \choose \Kt_0}}{{N_0 \choose K_0}}$. For this, we will need an additional proposition.

\begin{proposition}\label{prop:StirlingContinuity}
For any two nonnegative integers $a,b$, we have 
\begin{align}\label{eq:StirlingContinuity}
    \log\frac{(a+b)!}{a!} \le  b \log (a+b+1) + b + 1
\end{align}
\end{proposition}
\begin{proof}
By the Stirling approximation, for any positive integer $m$, we have 
\begin{align}\label{eq:stirlingapprox}
    \sqrt{2\pi}m^{m+1/2}e^{-m} \le m! \le e m^{m+1/2}e^{-m}.
\end{align}
We now use this to claim that when $a,b \ge 1$ 
\begin{align}
    \ln(a+b)! - \ln a! &\le \ln \left(e (a+b)^{a+b+1/2}e^{-(a+b)}\right) - \ln \left( \sqrt{2\pi}a^{a+1/2}e^{-a} \right) \nonumber \\
    &= \ln\frac{e}{\sqrt{2 \pi}} + (a+b+1/2)\ln(a+b) - (a+1/2) \ln a + a - (a+b) \nonumber \\
    &= \ln\frac{e}{\sqrt{2 \pi}} + b\ln(a+b) + (a+1/2)\ln(1+b/a) - b \nonumber \\
    &\le \ln\frac{e}{\sqrt{2 \pi}} + b\ln(a+b) + (a+1/2)\ln(1+b/a) - b \nonumber \\
    &\le  \ln\frac{e}{\sqrt{2 \pi}} + b\ln(a+b) + b/2a  \label{eq:lnoneplusx}\\
     &\le  \ln\frac{e}{\sqrt{2 \pi}} + b\ln(a+b) + b/2  \label{eq:apositive}
\end{align}
where~\eqref{eq:lnoneplusx} follows since for $x \ge 0, \ln(1+x) \le x$ and~\eqref{eq:apositive} follows since $a \ge 1$. When $a = b = 0$ and when $b = 0, a \ge 1$, the proposition is immediate. Finally, when $a = 0$ and $b \ge 1$, we have by the upper bound on $b!$ in~\eqref{eq:stirlingapprox} that
\begin{align}
    \ln b! \le \left(b + \frac{1}{2}\right) \ln b + 1 - b 
\end{align}
and after some algebraic manipulations we can see that the proposition holds in this case as well. 
\end{proof}

For convenience, define $\d_n := d_H(g^*(X^n),\gt(X^n))$. Note that 
\begin{align}
    \log \frac{{\Nt_0 \choose \Kt_0}}{{N_0 \choose K_0}} &= \log \frac{\Nt_0! K_0! (N_0-K_0)!}{N_0!\Kt_0! (\Nt_0 - \Kt_0)!} \nonumber \\
    &= \log \frac{\Nt_0!}{N_0!} + \log \frac{K_0!}{\Kt_0!} + \log \frac{(N_0-K_0)!}{(\Nt_0-\Kt_0)!}. \label{eq:sumoffactorials}
\end{align}
We will now bound each of the three terms in the RHS of~\eqref{eq:sumoffactorials}. 
We have 
\begin{align}
    \log \frac{\Nt_0!}{N_0!} &\le \log \frac{(N_0+\d_n)!}{N_0!} \label{eq:useHammUB} \\
    &\le \d_n\log(N_0+\d_n+1) + \d_n + 1 \label{eq:useStirlingCont} \\
    &\le \d_n\log(2n+1) + \d_n + 1 \label{eq:logrationfactorials} 
\end{align}
where~\eqref{eq:useHammUB} follows from Proposition~\ref{prop:HammDistNK}, ~\eqref{eq:useStirlingCont} follows from Proposition~\ref{prop:StirlingContinuity} and~\eqref{eq:logrationfactorials} follows since $N_0,\d_n \le n$. Using the same reasoning, we conclude
\begin{align}\label{eq:logratiokfactorials}
     \log \frac{K_0!}{\Kt_0!} \le \d_n\log(2n+1) + \d_n + 1.
\end{align}
and 
\begin{align}\label{eq:logrationkfactorials}
     \log \frac{(N_0-K_0)!}{(\Nt_0-\Kt_0)!} \le 2\d_n\log(3n+1) + 2\d_n + 1
\end{align}
where in~\eqref{eq:logrationkfactorials} we additionally use the fact that $|(N_0-K_0)-(\Nt_0-\Kt_0)| \le |N_0-\Nt_0| + |K_0-\Kt_0| \le  2\d_n$. Substituting~\eqref{eq:logrationfactorials}---~\eqref{eq:logrationkfactorials} into~\eqref{eq:sumoffactorials} yields 
\begin{align}\label{eq:nchoosekratioslogzero}
    \log \frac{{\Nt_0 \choose \Kt_0}}{{N_0 \choose K_0}} \le 4\d_n \log(3n+1) + 4\d_n + 3.
\end{align}
Similarly, we have
\begin{align}\label{eq:nchoosekratioslogone}
    \log \frac{{\Nt_1 \choose \Kt_1}}{{N_1 \choose K_1}} \le 4\d_n \log(3n+1) + 4\d_n + 3
\end{align}
and substituting~\eqref{eq:nchoosekratioslogzero} and~\eqref{eq:nchoosekratioslogone} into~\eqref{eq:ratiosofnchoosek} yields 
\begin{align}
      R_{n,P_X}(q_{\mix},f^*) &\le d \log\left(e^4n\right) + \Expt_{X^n,Y^n}\left[ 8\d_n \log(3n+1) + 8\d_n + 6\right] \nonumber \\
      &=  d \log\left(e^4n\right) + 8\log(6n+2)\Expt_{X^n,Y^n}[\d_n] + 6 \label{eq:SubstituteExpectedDelta}
\end{align}
Now, we have
\begin{align*}
    \Expt_{X^n,Y^n}[\d_n] =  \Expt_{X^n}\left[
    \sum_{i=1}^n \indic\{g^*(X_i) \neq \gt(X_i)\}\right] 
    = n \Prob(g^*(X_1) \neq \gt(X_1)) \le 1
\end{align*}
Since by design $d(\gt,g^*) = \Prob(\gt(X) \neq g^*(X)) \le 1/n$ where $X$ is distributed as $P_X$. Substituting this into~\eqref{eq:SubstituteExpectedDelta} yields 
\begin{align}
     R_{n,P_X}(q_{\mix},f^*) &\le (d+8)\log\left(e^4 n \right) + 6.
\end{align}

\section{Skipped proofs from Section~\ref{sec:lognRegret}}
\subsection{Proof of Proposition~\ref{prop:Halfspaceslogn}}
By Proposition~\ref{prop:ChainRuleMixture}, we have for a fixed $f^* = (a^*, \t_0^*, \t_1^*)$ 
\begin{align}
    R_{n,P_X}(q_{\mix},f^*) &= \Expt\left[\frac{p_{f^*}(Y^n|X^n)}{ \Expt_{F}[p_F(Y^n|X^n)]}\right] \nonumber \\
    &= \Expt \left[\log \frac{p_{f^*}(Y^n|X^n) }{\Expt_{A}[q_{A,\mix}(Y^n|X^n)]}\right]  \label{eq:ChainRuleHalfspaces} 
\end{align}
We will need the following claim. 
\begin{claim}\label{claim:continuityHS}
Let $a^* \in \sphere^{d-1}$ denote the function picked by the adversary, and $\d := \min_i|a^{*T}X_i|$. Then, for all $a \in \sphere^{d-1}$ such that $\|a-a^*\| < \d$, we have $q_{a,\mix}(Y^i|X^i) = q_{a^*, \mix}(Y^i|X^i)$. 
\end{claim}
\begin{proof}
Note that by definition of $q_{a,\mix}(Y^i|X^i)$, showing that
\[
 (g_a(X_1),\dotsc,g_a(X_n)) = (g_{a^{*}}( X_1),\dotsc,g_{a^*}(X_n))
\]
or equivalently that 
\begin{align}\label{eq:SphereLabelContinuity}
    (\sign(a^T X_1),\dotsc,\sign(a^T X_n)) = (\sign(a^{*T} X_1),\dotsc,\sign(a^{*T} X_n))
\end{align}
for all $\{a: \|a-a^*\| < \d\}$ suffices to prove the claim. Observe now that for all $a \in \sphere^{d-1}$ we have $\sign(a^T X_i) = \sign(a^{*T} X_i + (a-a^*)^T X_i)$, and if $\|a^* - a\| < \d$, we have $|(a-a^*)^T X_i| \le \|a-a^*\| < \d$, and therefore $\sign(w^T X_i) = \sign(w^{*T} X_i)$ for all $i = 1,\dotsc,n$ (since $|w^{*T} X_i| \ge \d$). This proves~\eqref{eq:SphereLabelContinuity} and consequently the claim.
\end{proof}
    
We now have 
\begin{align}
    q_{A,\mix}(Y^n|X^n) &\ge q_{A,\mix}(Y^n|X^n) \indic\{|a^*-A| < \d\} \nonumber \\
    &= q_{a^*,\mix}(Y^n|X^n) \indic\{|a^*-A| < \d\} \label{eq:useHScont}
\end{align}
where~\eqref{eq:useHScont} follows from Claim~\ref{claim:continuityHS}. Then, 
\begin{align}
    \Expt_A[q_{A, \mix}(Y^n|X^n)] &\ge \Expt_A[q_{a^*,\mix}(Y^n|X^n) \indic\{|a^*-A| < \d\}] \nonumber \\
    &= q_{a^*,\mix}(Y^n|X^n) \Expt_A[\indic\{|a^*-A| < \d\}] \nonumber \\
    &=  q_{a^*,\mix}(Y^n|X^n) \frac{\mathrm{Area}(\{\|a-a^*\| \le \d \}\cap \sphere^{d-1})}{\mathrm{Area}(\sphere^{d-1})} \label{eq:expectedIndicIsArea}
\end{align}
where~\eqref{eq:expectedIndicIsArea} follows since $A \sim \Unif[\sphere^{d-1}]$.

We now bound $\frac{\mathrm{Area}(\{\|a-a^*\| \le \d \}\cap \sphere^{d-1})}{\mathrm{Area}(\sphere^{d-1})}$ by a simple covering number argument explained next. Consider $\mathcal{N}(d,\d)$ to be a $\d-$covering of $\sphere^{d-1}$, which consists of the points $z_1,\dotsc,z_{|\mathcal{N}(d,\d)|}$. Then by definition of a covering, 
\[
\sphere^{d-1} = \cup_{i = 1}^{|\mathcal{N}(d,\d)|} \left(\{\|z-z_i\| \le \d\} \cap \sphere^{d-1}\right)
\]
and subsequently, 
\begin{align}
\mathrm{Area}(\sphere^{d-1}) &= \mathrm{Area}\left(\cup_{i = 1}^{|\mathcal{N}(d,\d)|} \left(\{\|z-z_i\| \le \d \}\cap \sphere^{d-1}\right)\right) \nonumber \\
&\le \sum_i \mathrm{Area}\left(\{\|z-z_i\| \le \d \}\cap \sphere^{d-1}\right) \nonumber \\
&= |\mathcal{N}(d,\d)| \mathrm{Area}(\{\|z-a^*\| \le \d \} \cap \sphere^{d-1}) \label{eq:CoveringNumBdOnarea}
\end{align}
where~\eqref{eq:CoveringNumBdOnarea} follows by symmetry of $\sphere^{d-1}$, which implies that any for each point $z \in \sphere^{d-1}$ the $\d-$neighbourhood is isomorphic. This establishes that 
\begin{align} \label{eq:AreaRatioCovering}
    \frac{\mathrm{Area}(\{\|a-a^*\| \le \d \}\cap \sphere^{d-1})}{\mathrm{Area}(\sphere^{d-1})} \ge \frac{1}{|\mathcal{N}(d,\d)|}.
\end{align}
Finally, we can show that when $\d \le 1$, 
\[
|\mathcal{N}(d,\d)| \le \left(\frac{3}{\d}\right)^d
\]
since $\mathcal{N}(d,\d) \le \mathcal{N}(\mathbb{B}_d,\d)$, the covering number of the unit ball, and $\mathcal{N}(\mathbb{B}_d,\d) \le \left(\frac{3}{\d}\right)^d$ ~\cite[Chapter 4]{vershynin2018high}. This implies that $\frac{\mathrm{Area}(\{\|a-a^*\| \le \d \}\cap \sphere^{d-1})}{\mathrm{Area}(\sphere^{d-1})} \ge \left(\frac{\d}{3}\right)^d$. Now, substituting this back in~\eqref{eq:expectedIndicIsArea} yields 
\begin{align}\label{eq:LBonQa}
  \Expt_A[q_{A,\mix}(Y^n|X^n)] \ge  q_{a^*,\mix}(Y^n|X^n) \left(\frac{\d}{3}\right)^d
\end{align}
and by substituting~\eqref{eq:LBonQa} into~\eqref{eq:ChainRuleHalfspaces}, we have 
\begin{align}\label{eq:substituteExptlogd}
    R_n(f^*, q_{\mix}) &\le \Expt\left[\log \frac{p_{f^*}(Y^n|X^n)}{q_{a^*,\mix}(Y^n|X^n)}\right] + d\Expt\left[\log \frac{3}{\d}\right].
\end{align}
We now consider $\Expt\left[\log \frac{1}{\d}\right]$. Recall that we have $\d = \min_{i} |a^{*T} X_i|$, where $X_i \sim \Unif(\sphere^{d-1})$ i.i.d. By symmetry, for any $a_1,a_2 \in \sphere^{d-1}$ we have 
\[
(|a_1^T X_1|,\dotsc,|a_1^T X_n|) \stackrel{(d)}{=} (|a_1^T X_1|,\dotsc,|a_2^T X_n|).
\]
In particular, choosing $a_1 = a^*$ and $a_2 = \begin{bmatrix} 1 & 0 \cdots  0 \end{bmatrix}$ we have 
\[
(|a^{*T} X_1|,\dotsc,|a^{*T} X_n|) \stackrel{(d)}{=}  (|X_{1,1}|,\dotsc,|X_{n,1}|) 
\]
where $X_{i,1}$ denotes the first co-ordinate of $X_i$. Now, for $X_i \sim \Unif(\sphere^{d-1})$, $X_{i,1} = 2Z-1$ where $Z \sim \mathrm{Beta}(d/2,d/2)$ (this follows directly from the formula for the surface area of the hyperspherical cap, see for example~\cite{li2011concise}). So, $X_{1,1},\dotsc,X_{n,1}$ are i.i.d. samples from a shifted and rescaled beta distribution. Thus, we can explicitly calculate $\Expt_{X^n}[- \log \d]$, which is simply $\Expt_{Z^n}[- \log (\min |2Z_i-1|)] = \Expt_{Z^n}[\max_i -\log|2Z_i-1|]$ where $Z_i \sim \mathrm{Beta}(d/2,d/2)$. We will next show that $\Expt[-\log \d] \le 2\ln(n) + o(1)$. 

Let $Z \sim \mathrm{Beta}(d/2,d/2)$, and $W := -\ln |2Z - 1|$. Since $Z \in [0,1]$, we have $W \ge 0$. Recalling that the density of $Z$ is $f_Z(z) = \frac{(z(1-z))^{d/2-1}}{\mathrm{B}(d/2,d/2)}, 0 \le z \le 1$, we can then calculate the density $f_W(w)$ as follows. We have, for any $w \ge 0$,
\begin{align*}
    1 -F_W(w) &= \Prob(W > w) \\
    &= \Prob(- \ln |2Z-1| > w) \\
     &= \Prob\left(\frac{1-e^{-w}}{2} < Z < \frac{1+e^{-w}}{2}\right) \\
     &= F_Z\left(\frac{1+e^{-w}}{2} \right) - F_Z\left(\frac{1-e^{-w}}{2}\right). \label{eq:DistOfW}
\end{align*}
Since $f_W(w) = \frac{d F_W(w)}{dw}$, taking derivative with respect to $w$ on both sides of~\eqref{eq:DistOfW} yields  
 \begin{align}
     f_W(w) &= \frac{d  F_Z\left(\frac{1-e^{-w}}{2}\right)}{dw} - \frac{d  F_Z\left(\frac{1+e^{-w}}{2}\right)}{dw} \nonumber\\
     &= \frac{1}{\mathrm{B}(d/2,d/2)}e^{-w}\left(\frac{1 - e^{-2w}}{4}\right)^{d/2-1} 
 \end{align}
Since $W$ is sub-exponential, we expect the scaling of $\Expt[\max\{ W_1,\dotsc,W_n\}]$ with $n$ to be $O(\log n)$ (i.e. similar to the dependence on $n$ of expected maximum for an exponential distribution). We next formalize this using a standard technique for bounding maximum of independent random variables. First, we provide a useful claim.

\begin{claim}\label{claim:DensityOfFappox}
For all $w \geq 0 $, we have $\frac{1}{\mathrm{B}(d/2,d/2)} \left(\frac{1 - e^{-2w}}{4}\right)^{d/2-1} \le c_d $ where $c_d := 2\sqrt{d}$, and subsequently $f_W(w) \le c_d e^{-w}$.
\end{claim}
\begin{proof}
Uses simple properties of the beta function and a Stirling approximation.
\end{proof}

For $n$ i.i.d. samples from $W$, denoted $W^n$, we next show that $\Expt[\max \{W_1,\dotsc,W_n\}] \le 2 \ln( 2 c_d n)$. We have 
\begin{align}
    \Expt[\max\{ W_1,\dotsc,W_n\}] &= 2\Expt\left[\ln \max \{e^{W_1/2},\dotsc, e^{W_n/2}\}\right] \\ 
    &\le 2 \ln\left(\Expt\left[\max  \{ e^{W_1/2},\dotsc, e^{W_n/2} \}\right]\right) \label{eq:JensenMaxIneq} \\
    &\le 2 \ln\left(\Expt\left[\sum_{i=1}^n e^{W_i}/2\right]\right) \\
    &= 2\ln\left(n\Expt\left[ e^{W_1/2}\right]\right) \\
    &= 2 \ln\left(n \int_0^{\infty} e^{w/2}f_W(w) dw \right) \\
    &\le 2 \ln\left(n \int_0^{\infty} e^{w/2}c_d e^{-w}dw \right) \label{eq:maximalIneqUB}\\ 
    &\le 2 \ln\left(c_d n  \int_0^{\infty} e^{-w/2} dw \right)=  2 \ln\left(2 c_d n\right).
\end{align}
where~\eqref{eq:JensenMaxIneq} follows from the Jensen inequality and~\eqref{eq:maximalIneqUB} follows from Claim~\ref{claim:DensityOfFappox}. Therefore, 
\begin{align}\label{eq:finalUBlogd}
\Expt[-\log \d] \le 2\ln n + 2 \ln (4 \sqrt{d}). 
\end{align}
Going back to~\eqref{eq:substituteExptlogd}, and using Lemma~\ref{lem:BlkKTUpBd} and~\eqref{eq:finalUBlogd} yields 
\begin{align}\label{eq:FinalRegHS}
    \max_{f \in \Fc} R_{n,P_X}(q_{\mix},f) \le (2 \ln 2d +1) \log n + d \log(48d) + \log\frac{\pi^2}{8}. 
\end{align}
\subsection{Proof of Proposition~\ref{prop:AxisAlignedRectslogn}}

The flow of this proof is almost the same as that of Proposition~\ref{prop:Halfspaceslogn}.
By Proposition~\ref{prop:ChainRuleMixture}, we have for a fixed $f^* = (\av^*, \bv^*, \t_0^*, \t_1^*)$ 
\begin{align}
    R_{n,P_X}(q_{\mix},f^*) &= \Expt\left[\frac{p_{f^*}(Y^n|X^n)}{ \Expt_{F}[p_F(Y^n|\Xv^n)]}\right]  \label{eq:ChainRuleAxisAlignedRects} 
\end{align}
Fix some $f^* = (\av^*,\bv^*,\t_0^*,\t_1^*)$. Now, recall that $\Xv_j \in \Real^d, j \in [n]$. Denote the $i-$th coordinate of $\Xv_j$ by $\Xv_{j,i}$. Now, given the $n$ real numbers $\Xv_{1,i},\dotsc,\Xv_{n,i}$ (i.e. the $i-$th coordinates of $\Xv_1,\dotsc,\Xv_n$) we can arrange these in order as $\Xv^i_{(1)} \le \Xv^i_{(2)} \le \dotsc \le \Xv^i_{(n)}$. Thus, $\Xv^i_{(1)},\dotsc,\Xv^i_{(n)}$ denote the order statistics of the $i-$th component of $\Xv_1,\dotsc,\Xv_n$. Now, clearly, there exist unique $k_i, l_i \in {0,\dotsc,n}, l_i \ge k_i$ such that $\Xv^i_{(k_i)} \le a_i^* \le  \Xv^i_{(k_i+1)}$ and $\Xv^i_{(l_i)} \le b_i^* \le  \Xv^i_{(l_i+1)}$. This holds for all $i \in [d]$. We then make the following claim.
\begin{claim}\label{claim:ContinuityOrderStats}
For any $\av, \bv$ that are such that for all $i \in [d], \Xv^i_{(k_i)} \le a_i \le  \Xv^i_{(k_i+1)}, \Xv^i_{(l_i)} \le b_i \le  \Xv^i_{(l_i+1)}$ and $a_i \le b_i$, we have $q_{\av,\bv,\mix}(Y^i|\Xv^i) = q_{\av^*,\bv^*,\mix}(Y^i|\Xv^i)$. 
\end{claim}
\begin{proof}
To prove this claim, note first that from the definition of $q_{\av,\bv,\mix}$ in~\eqref{eq:mixoverThetasOnly}, we have that if
\begin{align}\label{eq:continuityorderstatsSame}
(g_{\av^*,\bv^*}(\Xv_1),\dotsc,g_{\av^*,\bv^*}(\Xv_n)) = (g_{\av,\bv}(\Xv_1),\dotsc,g_{\av,\bv}(\Xv_n))
\end{align}
the claim holds. Then, for $\Xv_1$, we have $g_{\av^*,\bv^*}(\Xv_1) = \prod_{i=1}^d \indic\{a_i^* \le \Xv_{1,i} \le b_i^*\}$. Now, for any $i \in [d]$, since $\Xv^i_{(k_i)} \le a_i^* \le  \Xv^i_{(k_i+1)}$, this implies that if $\Xv^i_{(k_i)} \le a_i \le  \Xv^i_{(k_i+1)}$ we have that $\{j: \Xv_{j,i} \ge a^*_i\} = \{j: \Xv_{j,i} \ge a_i\}$ and therefore $\indic\{a_i^* \le \Xv_{1,i}\} = \indic\{a_i \le \Xv_{1,i}\}$. Similarly $\{j: \Xv_{j,i} \le b^*_i\} = \{j: \Xv_{j,i} \le b_i\}$ if $\Xv^i_{(l_i)} \le b_i \le  \Xv^i_{(l_i+1)}$, and consequently $\indic\{ \Xv_{1,i} \le b_i^*\} = \indic\{ \Xv_{1,i} \le b_i\}$. This implies that for any $\av,\bv$ satisfying the conditions of the claim, we have $\indic\{a_i^* \le \Xv_{1,i} \le b_i^*\} = \indic\{ a_i \le \Xv_{1,i} \le b_i\}$ for all $i \in [d]$, thereby implying that $g_{\av^*,\bv^*}(\Xv_1) = g_{\av,\bv}(\Xv_1)$. The same argument applied to $\Xv_2,\dotsc,\Xv_n$ implies~\eqref{eq:continuityorderstatsSame} and therefore the claim. 
\end{proof}

Now, we have 
\begin{align}
     q_{\Av,\Bv,\mix}(Y^n|\Xv^n) &\ge   q_{\Av,\Bv,\mix}(Y^n|\Xv^n) \prod_{i=1}^d \indic\{ \Xv^i_{(k_i)} \le A_i \le  \Xv^i_{(k_i+1)} \}\indic\{ \Xv^i_{(l_i)} \le B_i \le  \Xv^i_{(l_i+1)} \} \nonumber \\
     &= q_{\av^*,\bv^*,\mix}(Y^n|\Xv^n) \prod_{i=1}^d \indic\{ \Xv^i_{(k_i)} \le A_i \le  \Xv^i_{(k_i+1)} \}\indic\{ \Xv^i_{(l_i)} \le B_i \le  \Xv^i_{(l_i+1)}\} \label{eq:UseClaimOrderStats}
\end{align}
where~\eqref{eq:UseClaimOrderStats} follows from Claim~\ref{claim:ContinuityOrderStats}. Now, we have 
\begin{align}
    \Expt_F&\left[p_F(Y^n|\Xv^n)\right]
    \nonumber\\
    &= \Expt_{\Av,\Bv}\left[q_{\Av,\Bv,\mix}(Y^n|\Xv^n)\right] \nonumber \\
    &\ge \Expt_{\Av,\Bv}\left[q_{\av^*,\bv^*,\mix}(Y^n|\Xv^n) \prod_{i=1}^d \indic\{ \Xv^i_{(k_i)} \le A_i \le  \Xv^i_{(k_i+1)} \}\indic\{ \Xv^i_{(l_i)} \le B_i \le  \Xv^i_{(l_i+1)}\}\right] \label{eq:ContinuityArndActualAB} \\
     &= q_{\av^*,\bv^*,\mix}(Y^n|\Xv^n) \Expt_{\Av,\Bv}\left[\prod_{i=1}^d \indic\{ \Xv^i_{(k_i)} \le A_i \le  \Xv^i_{(k_i+1)} \}\indic\{ \Xv^i_{(l_i)} \le B_i \le  \Xv^i_{(l_i+1)}\}\right] \nonumber \\
     &=  q_{\av^*,\bv^*,\mix}(Y^n|\Xv^n) \prod_{i=1}^d \Expt_{A_i,B_i}\left[ \indic\{ \Xv^i_{(k_i)} \le A_i \le  \Xv^i_{(k_i+1)} \}\indic\{ \Xv^i_{(l_i)} \le B_i \le  \Xv^i_{(l_i+1)}\}\right] \label{eq:indepABcomponents} \\
     &\ge q_{\av^*,\bv^*,\mix}(Y^n|\Xv^n) \prod_{i=1}^d (\Xv^i_{(k_i+1)}-\Xv^i_{(k_i)})(\Xv^i_{(l_i+1)}-\Xv^i_{(l_i)})/2 \label{eq:AreaUnifAB}
\end{align}
 where~\eqref{eq:ContinuityArndActualAB} follows from~\eqref{eq:UseClaimOrderStats},~\eqref{eq:indepABcomponents} follows since the $(A_i,B_i)$ are all mutually independent, and~\eqref{eq:AreaUnifAB} follows since 
  \begin{align}
    \Expt_{A_i,B_i}&\left[ \indic\{ \Xv^i_{(k_i)} \le A_i \le  \Xv^i_{(k_i+1)} \}\indic\{ \Xv^i_{(l_i)} \le B_i \le  \Xv^i_{(l_i+1)}\}\right] \nonumber\\
    &=\Bigg\{ \begin{array}{lr}
      (\Xv^i_{(k_i+1)}-\Xv^i_{(k_i)})(\Xv^i_{(l_i+1)}-\Xv^i_{(l_i)})/2 & \text{for } l_i = k_i\\
       (\Xv^i_{(k_i+1)}-\Xv^i_{(k_i)})(\Xv^i_{(l_i+1)}-\Xv^i_{(l_i)}) & \text{for } l_i > k_i.
         \end{array}
   \end{align}
 By substituting~\eqref{eq:AreaUnifAB} into~\eqref{eq:ChainRuleAxisAlignedRects}, we get 
 \begin{align}
     R_n(q_{\mix}) \le
     \Expt\left[\log \frac{p_{\av^*,\bv^*,\t_0,\t_1}(Y^n|\Xv^n)}{q_{\av^*,\bv^*,\mix}(Y^n|\Xv^n)}\right] + \sum_{i=1}^d &\Expt_{\Xv^n} \left[\log \frac{2}{\Xv^i_{(k_i+1)}-\Xv^i_{(k_i)}}\right] \nonumber\\
     &+ \sum_{i=1}^d \Expt_{\Xv^n} \left[\log \frac{2}{\Xv^i_{(l_i+1)}-\Xv^i_{(l_i)}}\right] \label{eq:substituteExpectationsInReg}
 \end{align}
 Now, consider $\Expt_{\Xv^n} \left[\log \frac{2}{\Xv^i_{(k_i+1)}-\Xv^i_{(k_i)}}\right]$. Clearly, this quantity depends only on the $i-$th coordinates of $\Xv^n$, $\Xv_{1,i},\dotsc,\Xv_{n,i}$. Since $\Xv^n \sim \Unif[0,1]^d$ i.i.d, we can see that $\Xv_{1,i},\dotsc,\Xv_{n,i} \sim \Unif[0,1]$ i.i.d. 
 Now, it is known that for $Z^n \sim \Unif[0,1]$ i.i.d., $Z_{(k+1)}-Z_{(k)} \sim \mathrm{Beta}(1,n)$ for all $k \in \{0,\dotsc,n\}$. Moreover, for $Z' \sim \mathrm{Beta}(\a,\b)$, it can be shown that $\Expt[-\log Z'] \le \log(\a+\b)$. Using these two results, we can conclude that 
 \begin{align}\label{eq:BetaDistExpt}
 \Expt_{\Xv^n} \left[\log \frac{1}{\Xv^i_{(k_i+1)}-\Xv^i_{(k_i)}}\right],  \Expt_{\Xv^n} \left[\log \frac{1}{\Xv^i_{(l_i+1)}-\Xv^i_{(l_i)}}\right] \le \log(n+1), i \in [d].
 \end{align}
 Finally, using Lemma~\ref{lem:BlkKTUpBd} and~\eqref{eq:BetaDistExpt} in~\eqref{eq:substituteExpectationsInReg} yields 
 \begin{align}\label{eq:finalEqProofRects}
       \max_{f \in \Fc} R_{n,P_X}(q_{\mix},f) \le (2d+1)\log (n+1) + \log\frac{\pi^2}{8}
 \end{align}
 as required.
\section{Skipped Proofs from Section~\ref{sec:genupperbdpf}}

\subsection{Proof of Lemma~\ref{lem:regauxnoncausal}}

We have
\begin{align}
    \log&\frac{1}{\qt_{\mix}(Y_i|X^i,Y^{i-1},\Xt^n)} - \sum_{i=1}^n\log\frac{1}{p_f(Y_i|X_i)} \nonumber\\
    &\qquad= \log \frac{p_{f^*}(Y^n|X^n)}{\frac{1}{|\Pc(\Xt^n)|}\sum_{j=1}^{|\Pc(\Xt^n)|} \int_{0}^1 \int_{0}^1 p_{\gt_j,\t_0,\t_1}(Y^n|X^n) d\t_0d\t_1} \nonumber \\
    &\qquad \le d\log(en/d) + \log \frac{p_{f^*}(Y^n|X^n)}{\sum_{j=1}^{|\Pc(\Xt^n)|} \int_{0}^1 \int_{0}^1 p_{\gt_j,\t_0,\t_1}(Y^n|X^n) d\t_0d\t_1} \label{eq:sauershelahauxprob}.
\end{align}
So far, the construction and analysis of $\qt_{\mix}$ has paralleled the analysis of the mixture $q_{\mix}$ in Section~\ref{subsec:noncausalsideinfo}. There, the next step was to claim that since $\exists j \in \left[|\Pc_n(X^n)|\right]$ such that $g_j(X^n) = g^*(X^n)$, invoking Lemma~\ref{lem:BlkKTUpBd} yielded an $O(\log n)$ upper bound for the second term in~\eqref{eq:sauershelahauxprob}. Unfortunately we cannot claim the same in the current case. However, we can claim that $\exists \tilde{j} \in \left[|\Pc_n(\Xt^n)|\right]$ such that 
\[
\gt_{\tilde{j}}(\Xt^n) = g^*(\Xt^n).
\]
Since $d_H(\gt_{\tilde{j}}(\Xt^n), g^*(\Xt^n)) = 0$ and $\Xt^n \stackrel{(d)}{=} X^n$, we intuitively expect  $d_H(\gt_{\tilde{j}}(X^n), g^*(X^n))$ to not be too large. We now quantify this intuition more precisely. For brevity, denote 
\[
\gt := \gt_{\tilde{j}}.
\]
We have from~\eqref{eq:sauershelahauxprob}
\begin{align}
    \log\frac{1}{\qt_{\mix}(Y_i|X^i,Y^{i-1},\Xt^n)}& - \sum_{i=1}^n\log\frac{1}{p_f(Y_i|X_i)}
    \nonumber\\
    &\le d\log(en/d) + \log \frac{p_{f^*}(Y^n|X^n)}{ \int_{0}^1 \int_{0}^1 p_{\gt,\t_0,\t_1}(Y^n|X^n) d\t_0d\t_1} \label{eq:mismatchedganalysis} \\
    &\le d\log(en/d) + 8\log(6n+2)d_H(g^*(X^n),\gt(X^n)) + 6 \label{eq:subExpectedHammDist}
\end{align}
Where to get from~\eqref{eq:mismatchedganalysis} to~\eqref{eq:subExpectedHammDist} we follow the exact same steps employed in the proof of Lemma~\ref{lem:PxKnownReg} from~\eqref{eq:SubVcCoveringTwo} to~\eqref{eq:SubstituteExpectedDelta}.

We now focus on  $d_H(g^*(X^n),\gt(X^n))$ and establish that
\begin{align}
\Expt_{X^n,\Xt^n}[d_H(g^*(X^n),\gt(X^n))] &\le 2C\sqrt{dn} \label{eq:ExpectedSupremumVCclass} \\
d_H(g^*(X^n),\gt(X^n)) &\le 2C\sqrt{dn} + 2\sqrt{2n \log \frac{2}{\d}} \text{  , with probability  }\ge 1 - \d \label{eq:SupremumVCclassConc}
\end{align}
for an absolute constant $C \le 250$.

For any $g_1, g_2 \in \Gc$ define 
\begin{align*}
    \D_n(g_1,g_2) &:= \frac{1}{n}d_H(g_1(X^n),g_2(X^n)) \\ 
    \Dt_n(g_1,g_2) &:= \frac{1}{n}d_H(g_1(\Xt^n),g_2(\Xt^n))  \\
    \D(g_1,g_2) &:= \Prob(g_1(X) \neq g_2(X))
    \end{align*}
    for $X \stackrel{(d)}{=} X_1 \stackrel{(d)}{=} \Xt_1$. Recall that $\Dt_n(\gt,g^*) = 0$ by design, and $\D(g_1,g_2) = \Expt_{X^n}[\D_n(g_1,g_2)] = \Expt_{\Xt^n}[\Dt_n(g_1,g_2)]$. We then have
\begin{align}
    \D_n(g^*,\gt) &= \D_n(g^*,\gt) - \Dt_n(g^*,\gt) \nonumber\\
    &\le \sup_{g_1,g_2 \in \Gc}\left|\D_n(g_1,g_2) - \Dt_n(g_1,g_2)\right| \nonumber \\
    &\le \sup_{g_1,g_2 \in \Gc}\left|\D_n(g_1,g_2) - \D(g_1,g_2)\right| +  \sup_{g_1,g_2 \in \Gc}\left| \Dt_n(g_1,g_2) - \D(g_1,g_2)\right|. \label{eq:suptriangleineq}
\end{align}
We first establish~\eqref{eq:ExpectedSupremumVCclass}. Taking expectations on both sides of~\eqref{eq:suptriangleineq}.
\begin{align}
    \Expt_{X^n,\Xt^n}&[\D_n(g^*,\gt)] \nonumber\\
    &\le \Expt_{X^n,\Xt^n}\left[\sup_{g_1,g_2 \in \Gc}\left|\D_n(g_1,g_2) - \D(g_1,g_2)\right| +  \sup_{g_1,g_2 \in \Gc}\left| \Dt_n(g_1,g_2) - \D(g_1,g_2)\right|\right] \nonumber \\
    &= 2\Expt_{X^n}\left[\sup_{g_1,g_2 \in \Gc}\left|\D_n(g_1,g_2) - \D(g_1,g_2)\right|\right] \label{eq:samedistxtilde}
\end{align} 
where~\eqref{eq:samedistxtilde} follows by linearity of expectation and since $X^n \stackrel{(d)}{=} \Xt^n$. 
Finally, we note that 
\begin{align*}
\D_n(g_1,g_2) = \frac{d_H(g_1(X^n),g_2(X^n))}{n} &= \frac{1}{n}\sum_{i=1}^n \indic\{g_1(X_i) \neq g_2(X_i)\} \\
\D(g_1,g_2) &= \Expt[\indic\{g_1(X) \neq g_2(X)\}]
\end{align*}
and the class of boolean functions $\{x \mapsto \indic\{g_1(x) \neq g_2(x)\}, (g_1,g_2) \in \Gc \times \Gc\}$ has VC dimension $\le 2d$. Thus, we can now invoke~\cite[Theorem 8.3.23]{vershynin2018high},~\cite[Theorem 13.7]{Boucheron--Lugosi--Massart13} to claim that 
\begin{align}\label{eq:SupremumGSquareUB}
    \Expt_{X^n}\left[\sup_{g_1,g_2 \in \Gc}|\D_n(g_1,g_2) - \D(g_1,g_2)|\right] \le C\sqrt{\frac{d}{n}} 
\end{align}
for a universal constant $C \le 250$. Consequently, taking expectations on both sides of~\eqref{eq:subExpectedHammDist} and substituting~\eqref{eq:SupremumGSquareUB}, followed by a supremum over $f^*$ and $P_X$ yields 
\begin{align}
    \tilde{R}_n(\qt_{\mix}) \le d\log (en/d) + 16C\sqrt{nd}\log(6n+2)
\end{align}
as required. 

To establish~\eqref{eq:SupremumVCclassConc}, we invoke Theorem 12.1 of~\cite{Boucheron--Lugosi--Massart13} to assert 
\begin{align}
    \sup_{g_1,g_2 \in \Gc}\left|\D_n(g_1,g_2) - \D(g_1,g_2)\right| \le \Expt \left[\sup_{g_1,g_2 \in \Gc}\left|\D_n(g_1,g_2) - \D(g_1,g_2)\right|\right] + \sqrt{\frac{2}{n} \log \frac{2}{\d}}
\end{align}
with probability $1-\d/2$. The same high-probability bound for  $\sup_{g_1,g_2 \in \Gc}\left|\Dt_n(g_1,g_2) - \D(g_1,g_2)\right|$ along with a union bound and~\eqref{eq:SupremumGSquareUB} yields~\eqref{eq:SupremumVCclassConc}. Substituting~\eqref{eq:SupremumVCclassConc} into~\eqref{eq:subExpectedHammDist} yields the second part of the lemma.

\subsection{Proof of Theorem~\ref{thm:genupperbd}}

We have 
\begin{align}
    \sum_{i=1}^n \log &\frac{1}{q^*(Y_i|X^i,Y^{i-1})} - \sum_{i=1}^n \log\frac{1}{p_{f}(Y_i|X_i)} \nonumber\\
    &\le  \sum_{i=2}^n \log \frac{1}{q^*(Y_i|X^i,Y^{i-1})} - \sum_{i=2}^n \log\frac{1}{p_{f}(Y_i|X_i)} + 1 \nonumber \\
    &= \sum_{j=1}^{\log n} \left[\sum_{i=2^{j-1}+1}^{2^j} \log \frac{1}{q^*(Y_i|X^i,Y^{i-1})} - \sum_{i=2^{j-1}+1}^{2^j} \log\frac{1}{p_{f}(Y_i|X_i)} \right]  \label{eq:SplittingIntoEpochs}
\end{align}
Taking expectation on both sides of~\eqref{eq:SplittingIntoEpochs}, we have 
\begin{align}
    R_{n,P_X}(q^*,f) &\le \sum_{j=1}^{\log n}  \Expt \left[\sum_{i=2^{j-1}+1}^{2^j} \log \frac{1}{q^*(Y_i|X^i,Y^{i-1})} - \hspace{-0.8em}\sum_{i=2^{j-1}+1}^{2^j} \log\frac{1}{p_{f}(Y_i|X_i)} \right] + 1 \nonumber\\
    &\le \sum_{j=1}^{\log n} \tilde{R}_{2^{j-1}}(\qt_{\mix}) \label{eq:SimilarityToAuxProb} 
\end{align}
where, recall, the expectations in the first inequality are w.r.t. $X_{2^{j-1}+1}^{2^j}, Y_{2^{j-1}+1}^{2^j}, X_1^{2^{j-1}}$ and~\eqref{eq:SimilarityToAuxProb} follows since $q^*$ is exactly $\qt_{\mix}$. Using Lemma~\ref{lem:regauxnoncausal}, we have for any $n' \ge 2$ 
\[
\tilde{R}_{n'}(\qt_{\mix}) \le d \log (en'/d) + 16C\sqrt{dn'}\log(6n'+2) \le  d \log n' + 64 C\sqrt{dn'}\log(n')
\]
and therefore, from~\eqref{eq:SimilarityToAuxProb}
\begin{align}
     R_{n,P_X}(q^*,f) &\le \sum_{j=2}^{\log n} \left(d(j-1) + 64 C\sqrt{d} 2^{(j-1)/2}(j-1)\right)  + 2 \nonumber \\
     &\le d(\log n)^2 + 64 C\sqrt{d} \int_{1}^{\log n + 1} x2^{x/2} dx + 2 \label{eq:SumToIntegral}\\
     &\le d(\log n)^2 + 125 C\sqrt{dn}\log(2n) + 2 \label{eq:SumOfEpochRegrets} 
\end{align}
and finally taking supremum over $f$ and $P_X$ concludes the first part of the proof.

For the second part, we have from Lemma~\ref{lem:regauxnoncausal} for any $j \ge 2$,
\begin{align}
\sum_{i=2^{j-1}+1}^{2^j} \log \frac{1}{q^*(Y_i|X^i,Y^{i-1})} &- \sum_{i=2^{j-1}+1}^{2^j} \log\frac{1}{p_{f}(Y_i|X_i)}  \nonumber\\
&\le d(j-1) + 64 2^{(j-1)/2}(j-1)\left(C \sqrt{d} + \sqrt{2 \log \frac{2\log n}{\d}} \right) 
\end{align}
with probability $1 - \d/\log n$. Then from~\eqref{eq:SplittingIntoEpochs}, a union bound and the calculations in~\eqref{eq:SumToIntegral} and~\eqref{eq:SumOfEpochRegrets} we have 
\begin{align}
     \sum_{i=1}^n \log &\frac{1}{q^*(Y_i|X^i,Y^{i-1})} - \sum_{i=1}^n \log\frac{1}{p_{f}(Y_i|X_i)} \nonumber\\
     &\qquad\le d(\log n)^2 + 125 C\sqrt{dn}\log(2n) \left(C \sqrt{d} + \sqrt{2 \log \frac{2\log n}{\d}} \right) + 2
\end{align}
with probability $\ge 1-\d$ as required.
\section{Skipped Proofs from Section~\ref{sec:lowerbd}}

\subsection{Proof of Lemma~\ref{lem:MinimaxLBRegretFinite}}

We have from~\eqref{eq:CondRedCapThmOne}
\begin{align}\label{eq:RegretLBCapacity}
    R_n \ge \max_{P_X,P_F} I(\T_0,\T_1,G; Y^n|X^n)
\end{align}
and therefore a lower bound on $I(\T_0,\T_1,G; Y^n|X^n)$ for any choice of $P_X$ and $P_F$ provides a lower bound on $R_n$. We will choose $P_X$ to be the uniform distribution on $\{1,\dotsc,m\}$ so that 
\[
X \sim \Unif\left([m]\right).
\]
Consider now the following distribution $P_F$ over the hypothesis class $F = (\T_0,\T_1,G)$ that has
\begin{align}
    (\T_0, \T_1) &\sim \Unif\left(\t_0,\t_1 \in [0,1]\times [0,1] \cap \left\{\t_1-\t_0\ge 1/2\right\}\right) \label{eq:DistOfThetas}\\
     G  &\perp \!\!\! \perp (\T_0,\T_1) \text{ and } G \sim \Unif\left\{2^{[m]}\right\} \label{eq:DistOfG}
\end{align}
We then have 
\begin{align}
    I&(\T_0,\T_1,G; Y^n|X^n)\nonumber\\
    &= H(\T_0,\T_1,G|X^n) - H(\T_0,\T_1,G|X^n,Y^n) \nonumber\\
    &= h(\T_0,\T_1) + H(G) - H(\T_0,\T_1,G|X^n,Y^n)  \label{eq:indepOfHypotheses} \\
    &\ge  h(\T_0, \T_1) + H(G) - H(G|X^n,Y^n) - H(\T_0|G,X^n,Y^n) - H(\T_1|G,X^n,Y^n) \label{eq:chainruleEnt} \\
    &= 3 + m - H(G|X^n,Y^n) - h(\T_0|G,X^n,Y^n) - h(\T_1|G,X^n,Y^n) \label{eq:marginalEnt}
\end{align}
where~\eqref{eq:indepOfHypotheses} follows since the $G  \perp \!\!\! \perp (\T_0, \T_1)$ and both $(\T_0,\T_1), G  \perp \!\!\! \perp X^n$,~\eqref{eq:chainruleEnt} follows from the chain rule of entropy and because conditioning reduces entropy, and~\eqref{eq:marginalEnt} follows since by the distribution of $(\T_0,\T_1)$ and $G$ in~\eqref{eq:DistOfThetas},~\eqref{eq:DistOfG}, we have $(\T_0,\T_1)$ is uniform over a set with area $\frac{1}{8}$, and $G \sim 
\Unif(\Gc)$ with $|\Gc| = 2^m$. 


We also have 
\begin{align}
    h(\T_0|X^n,Y^n,G) &= \sum_{g \in \Gc} h(\T_0|X^n,Y^n,G=g)\Prob(G = g)  \nonumber \\
   &= \frac{1}{2^m}  \sum_{g \in \Gc} h(\T_0|X^n,Y^n,G=g) \label{eq:ThZeroEntropy}
\end{align}
Now, define the estimator 
\begin{align}\label{eq:thetahatDef}
    \Th_0(X^n,Y^n,g) = \frac{  \sum_{i=1}^n  \indic\{g(X_i) = 0, Y_i = 1\} + 1/2 }{ \sum_{i=1}^n  \indic\{g(X_i) = 0\} + 1}.
\end{align}
Defining $N_0 := \sum_{i=1}^n  \indic\{g(X_i) = 0\}$ and $K_0 =  \sum_{i=1}^n  \indic\{g(X_i) = 0, Y_i = 1\}$, we have 
\[
\Th_0  = \frac{K_0 + 1/2}{N_0 + 1}.
\]
Now, going back to~\eqref{eq:ThZeroEntropy}, we have 
\begin{align}
    \frac{1}{2^m}  \sum_{g \in \Gc} h(\T_0|X^n,Y^n,G=g) &\le   \frac{1}{2^m}  \sum_{g \in \Gc} h(\T_0|\Th_0,g) \label{eq:DataProcIneq} \\
    &= \frac{1}{2^m}  \sum_{g \in \Gc} h(\T_0 - \Th_0|\Th_0,g) \nonumber \\
    &\le \frac{1}{2^m}  \sum_{g \in \Gc} h(\T_0 - \Th_0|g) \nonumber \\
    &\le \frac{1}{2^m}  \sum_{g \in \Gc} \frac{1}{2}\log(2 \pi e \mathrm{Var}(\T_0 - \Th_0|g)) \label{eq:GaussianMaxEnt}\\
    &\le \frac{1}{2^m}  \sum_{g \in \Gc} \frac{1}{2}\log(2 \pi e \Expt[(\T_0 - \Th_0)^2|g])  \label{eq:MomentIneq}
\end{align}
where~\eqref{eq:DataProcIneq} follows from the data processing inequality,~\eqref{eq:GaussianMaxEnt} follows since the Gaussian random variable of a given variance maximizes entropy, and~\eqref{eq:MomentIneq} follows since for any random variable $Z, \mathrm{Var}[Z] \le \Expt[Z^2]$. 

We now have 
\begin{align}
\Expt[(\T_0 - \Th_0)^2|g] &= \Expt_{\T_0,X^n,Y^n|g} (\T_0 - \Th_0)^2 \nonumber \\
&= \Expt_{\T_0,X^n,Y^n|g}\left(\T_0 -   \frac{ K_0 + 1/2 }{ N_0 + 1} \right)^2  \nonumber \\
&= \Expt_{\T_0,N_0,K_0|g}\left(\T_0 -   \frac{ K_0 + 1/2  }{ N_0 + 1} \right)^2  \nonumber \\
&= \Expt_{\T_0|g}   \Expt_{N_0|\T_0,g}  \Expt_{K_0|N_0,\T_0,g} \left[  \left( \T_0 -   \frac{ K_0 + 1/2  }{ N_0 + 1} \right)^2   \Big| N_0, \T_0 \right] \label{eq:BinomCalc}
\end{align} 
Since 
\[
K_0|N_0,\T_0,g \sim \Bin(N_0,\T_0)
\]
we can calculate 
\begin{align}
    \Expt_{K_0|N_0,\T_0,g} \left[  \left( \T_0 -   \frac{ K_0 + 1/2  }{ N_0 + 1} \right)^2   \Big| N_0, \T_0 \right] &= \frac{(\T_0-1/2)^2 + N_0\T_0(1-\T_0)}{(N_0+1)^2} \nonumber \\ 
    &\le \frac{1}{4(N_0+1)} \label{eq:QuadIneq} 
\end{align}
where~\eqref{eq:QuadIneq} follows since $x(1-x) \le \frac{1}{4}, (x-1/2)^2 \le \frac{1}{4}$ for $x \in [0,1]$. Substituting~\eqref{eq:QuadIneq} back into~\eqref{eq:BinomCalc} we obtain 
\begin{align}
    \Expt(\T_0 - \Th_0)^2 &\le \Expt_{\T_0|g}   \Expt_{N_0|\T_0,g}\left[ \frac{1}{4(N_0+1)}  \right] \nonumber \\
    &= \Expt_{N_0|g} \left[ \frac{1}{4(N_0+1)}  \right] \label{eq:IndepNZero} 
\end{align}
where~\eqref{eq:IndepNZero} follows since $N_0|g \dperp \T_0$ with distribution $N_0 \sim \Bin\left(n, \sum_{i:g(i)=1}\Prob(X = i)\right)$. Defining 
\[
p_g := \sum_{i:g(i)=1}\Prob(X = i),
\]
we can  the see that when $p_g \neq 0$ by a simple binomial calculation 
\begin{align}\label{eq:inverseBinExpt}
\Expt_{N_0|g} \left[ \frac{1}{4(N_0+1)}  \right] = \frac{1-(1-p_g)^{n+1}}{4(n+1)p_g} 
\end{align}
and $\Expt_{N_0|g} \left[ \frac{1}{4(N_0+1)}  \right] = \frac{1}{4}$ when $p_g = 0$. 
Now, we have 
\begin{align}
     h(\T_0|X^n,Y^n,G) &=  \frac{1}{2^m}  \sum_{g \in \Gc} h(\T_0|X^n,Y^n,G=g) \nonumber \\
     &\le \frac{1}{2^m}  \sum_{g \in \Gc}  \frac{1}{2}\log\left( 2 \pi e \Expt_{N_0|g} \left[ \frac{1}{4(N_0+1)}  \right]\right) \label{eq:VarCalc} \\
     &= \frac{1}{2}\log( \pi e/2) + \sum_{g \in \Gc} \frac{1}{2}\log \left(\Expt_{N_0|g}\left[ \frac{1}{N_0+1}  \right]\right) \label{eq:HZeroUpperBd}
\end{align}
where~\eqref{eq:VarCalc} follows from~\eqref{eq:MomentIneq}. 

In the exact same way, we can upper-bound $h(\T_1|X^n,Y^n,G)$ as 
\begin{align}\label{eq:HOneUpperBd}
     h(\T_1|X^n,Y^n,G)  \le \frac{1}{2}\log( \pi e/2) + \sum_{g \in \Gc} \frac{1}{2}\log \left(\Expt_{N_1|g}\left[ \frac{1}{N_1+1}  \right]\right).
\end{align}
From~\eqref{eq:HZeroUpperBd} and~\eqref{eq:HOneUpperBd} we get 
\begin{align}
    h(\T_0|X^n,Y^n,G) &+ h(\T_1|X^n,Y^n,G) \nonumber\\
    &\le \log(\pi e/2) + \sum_{g \in \Gc} \frac{1}{2} \log  \left(\Expt_{N_0|g}\left[ \frac{1}{N_0+1}  \right] \Expt_{N_1|g}\left[ \frac{1}{N_1+1}  \right]  \right) 
\end{align}
Now, from~\eqref{eq:inverseBinExpt} we have, when $p_g \neq 0, 1$ \begin{align}
    \Expt_{N_0|g}\left[ \frac{1}{N_0+1}  \right] \Expt_{N_1|g}\left[ \frac{1}{N_1+1}  \right]  &= \frac{1-(1-p_g)^{n+1}}{(n+1)p_g} \cdot 
    \frac{1-p_g^{n+1}}{(n+1)(1-p_g)} \nonumber\\
    &\le \frac{1}{n+1} \label{eq:invBinProdBd}
\end{align}
where~\eqref{eq:invBinProdBd} follows from noting that the function $\frac{1-(1-x)^{n+1}}{x}\cdot\frac{1-x^{n+1}}{1-x} \le n+1$ for all $0 < x < 1$. Moreover, when $p_g$ is either 0 or 1 we have $\Expt_{N_0|g}\left[ \frac{1}{N_0+1}  \right] \Expt_{N_1|g}\left[ \frac{1}{N_1+1}  \right]  = \frac{1}{n+1}$, and putting the aforementioned two cases together we have 
\begin{align}
   h(\T_0|X^n,Y^n,G) + h(\T_1|X^n,Y^n,G) &\le \log(\pi e/2) + \sum_{g \in \Gc} \frac{1}{2} \log  \left(\frac{1}{n+1}\right)  \nonumber \\
   &\le \log(\pi e/2) - \frac{1}{2}\log(n+1) \label{eq:ThZeroThOneEntSum}
\end{align}

Substituting the bound~\eqref{eq:ThZeroThOneEntSum} into~\eqref{eq:marginalEnt} yields 
\begin{align}\label{eq:putinerrorprobfano}
    I(\T_0,\T_1,G;Y^n|X^n) &\ge m + \log(n+1) - H(G|X^n,Y^n) - \log(4\pi e).
\end{align}
Now, we have for any estimator $\widehat{G}(X^n,Y^n)$ of $G$, 
\begin{align}
    H(G|X^n,Y^n) &\le H(G|\widehat{G}(X^n,Y^n)) \label{eq:dataprocineq}\\
    &\le \Prob(G \neq \widehat{G}(X^n,Y^n))m + 1 \label{eq:FanoIneq}
\end{align}
where~\eqref{eq:dataprocineq} follows by the data processing inequality, and~\eqref{eq:FanoIneq} follows from the Fano inequality~\cite{10.5555/1146355}. 
We now provide an estimator $\widehat{G}(X^n,Y^n)$ for which the error probability $\Prob(G \neq \widehat{G}(X^n,Y^n)) = o(1)$. Given $X^n,Y^n$, we  define
\begin{align}\label{eq:pHatDef}
  \ph_{l} := \frac{\sum_{i=1}^{n} \indic\{X_i = l, Y_i = 1\} + 1/2}{\sum_{i=1}^{n} \indic\{X_i = l\} + 1}, \text{ } l \in \{1,\dotsc,m\}.
\end{align}
Let $\ph_{\min} := \min_l \ph_l$ and $\ph_{\max} := \max_l \ph_l$. The estimator $\Gh(X^n, Y^n) \in \Gc$ is then defined as
 \[
   \Gh(l) = \Bigg\{\begin{array}{lr}
        0 & \text{if } \ph_l \le \frac{\ph_{\max} + \ph_{\min}}{2}\\
       1 & \text{otherwise. }
        \end{array}
  \]
The probability of error of this estimator can now be bounded as follows.

\begin{lemma}\label{lem:GhatProbError}
We have 
\begin{align}\label{eq:GhatProbError}
    \Prob(\Gh(X^n,Y^n) \neq G) \le \frac{2}{2^m} + \left(1-\frac{2}{2^m}\right)2\sqrt{e} m e^{-3n/100m}. 
\end{align}
\end{lemma}
The proof of Lemma~\ref{lem:GhatProbError} is provided in the next subsection of Appendix D. 

Using Lemma~\ref{lem:GhatProbError} in~\eqref{eq:FanoIneq} and substituting this into~\eqref{eq:putinerrorprobfano}, since $\frac{2m}{2^m} \le 1$, we have 
\begin{align}
     I(\T_0,\T_1,G;Y^n|X^n) &\ge m + \log(n+1) - 2\sqrt{e} m^2 e^{-3n/100m}  - \log(\pi e)
\end{align}
as required.
\subsection{Proof of Lemma~\ref{lem:GhatProbError}}
We will denote $\Gh(X^{n},Y^{n})$ simply by $\Gh$ for convenience. 

Let $g = 0$ and $g = 1$ denote the all-0 and all-1 functions respectively (i.e. $g(x) = 0/1$ for all $x \in [m]$). We have 
\begin{align}
    \Prob(\Gh \neq G) &= \frac{1}{2^m} \sum_{g \in \Gc} \Prob(\Gh \neq g|G = g) \nonumber\\
    &= \frac{1}{2^m} \left(\Prob(\Gh \neq 0|G = 0) +  \Prob(\Gh \neq 1|G = 1)\right) +  \frac{1}{2^m} \sum_{g \in \Gc} \Prob(\Gh \neq g|G = g) \nonumber \\
    &\le \frac{2}{2^m} +  \frac{1}{2^m} \sum_{g \in \Gc \setminus \{g = 0, g = 1\}} \Prob(\Gh \neq g|G = g) \label{eq:probErrorChainRule}
\end{align}
Now, consider $\Prob(\Gh \neq g|G = g)$ for $g \neq 0,1$ identically. Since 
\[
\Prob(\Gh \neq g|G = g) = \Expt_{\T_0,\T_1}[\Prob(\Gh \neq g|G = g,\T_0,\T_1)],
\]
showing that for a \emph{fixed} $(g,\t_0,\t_1)$ with $g \neq 0,1$ identically and $\t_1-\t_0 \ge \frac{1}{2}$, with $X_i \sim \Unif\{[m]\}$ i.i.d. and $Y_i|(X_i = l) \sim \Ber(\t_{g(l)}), i \in [n]$, $\Prob(\Gh \neq g) \le 2\sqrt{e}m  e^{-3n/100 m}$ suffices to prove the lemma (recall that the $\t_1-\t_0 \ge \frac{1}{2}$ condition arises due to the choice of $P_H$ and more specifically the distribution of $(\T_0,\T_1)$ in~\eqref{eq:DistOfThetas}, which has zero density over the region $\t_1-\t_0 < \frac{1}{2}$). We now prove this statement.  

We claim that 
\begin{align}\label{eq:eventInclusion}
    \left\{\cap_{l = 1}^m (|\ph_l - \t_{g(l)}|\le 1/8) \right\} \subseteq \left\{\Gh = g\right\}.
\end{align}
To see this, note that if the event $\left\{\cap_{l = 1}^m (|\ph_l - \t_{g(l)}|\le 1/8) \right\}$ occurs, we have $\t_1 - 1/8 \le \ph_{\max} \le \t_1 + 1/8$ and $\t_0 - 1/8 \le \ph_{\min} \le \t_0 + 1/8$ (recall that there is at least one $l$ such that $g(l) = 0$, and similarly at least one $l$ such that $g(l) = 1$) and subsequently, adding these two inequalities,
\begin{align}\label{eq:BoundaryBounds}
\frac{\t_0 + \t_1}{2} - 1/8 \le \frac{\ph_{\max} + \ph_{\min}}{2} \le \frac{\t_0 + \t_1}{2} + 1/8
\end{align}
But, since $\t_1 - \t_0 \ge 1/2$, we have $\t_0 + 1/8 \le \frac{\t_0 + \t_1}{2} - 1/8 $ and similarly $\t_1 - 1/8 \ge \frac{\t_0 + \t_1}{2} + 1/8$. This, together with~\eqref{eq:BoundaryBounds} implies that 
\[
\t_0 + 1/8 \le \frac{\ph_{\max} + \ph_{\min}}{2} \le \t_1 - 1/8
\]
Since the event  $\left\{\cap_{l = 1}^m (|\ph_l - \t_{g(l)}|\le 1/8) \right\}$ occurring implies that if $g(l) = 0, \ph_l \le \t_0 + 1/8$, which implies that in this case $\ph_l \le \frac{\ph_{\max} + \ph_{\min}}{2}$ and so $\Gh(l) = g(l) = 0$. Similarly, when $g(l) = 1$, $\Gh(l) = g(l) = 1$. 

Going back to~\eqref{eq:eventInclusion}, we have
\begin{align}
    \Prob\left(\cap_{l = 1}^m |\ph_l - \t_{g(l)}|\le 1/8 \right) &\le \Prob(\Gh = g) \nonumber \\
    \implies \Prob(\Gh \neq g) &\le  \Prob\left(\cup_{l = 1}^m |\ph_l - \t_{g(l)}|> 1/8\right)  \nonumber \\
    \implies \Prob(\Gh \neq g) &\le  \sum_{l = 1}^{m} \Prob\left(|\ph_l - \t_{g(l)}|> 1/8\right) \label{eq:UnionBd}
\end{align}
where~\eqref{eq:UnionBd} follows from the union bound. Consider now $\Prob\left(|\ph_m - \t_{g(m)}|> 1/8 \right)$. Without loss of generality, we may assume that $g(m) = 1$. Introducing the notation\footnote{This notation is independent of and not to be confused with the definitions of $N_0, K_0, N_1$ and $K_1$ in the proof of Lemma~\ref{lem:MinimaxLBRegretFinite}.}
\begin{align}
     N_l &:= \sum_{i=1}^{n} \indic\{X_i = l\},  \text{ } l \in \{1,\dotsc,m\} \label{eq:nlPrimeDef} \\
  K_l &:= \sum_{i=1}^{n} \indic\{Y_i = 1, X_i = l\},  \text{ } l \in \{1,\dotsc,m\}. \label{eq:klPrimeDef} 
\end{align}
we have 
\begin{align*}
\Prob\left(|\ph_m - \t_1|> 1/8 \right) &= \Prob\left(\left|\frac{K_m + 1/2}{N_m + 1} - \t_1\right|> 1/8\right) \nonumber\\
&= \Expt_{N_m}\left[ \Prob\left(\left|\frac{K_m + 1/2}{N_m + 1} - \t_1\right|> 1/8\Big| N_m\right) \right].
\end{align*}
Recalling now that $K_m|N_m \sim \Bin(N_m,\t_1)$ we have, by a slight variation on the Hoeffding inequality 
\begin{align}\label{eq:HoeffdingIneq}
    \Prob\left(\left|\frac{K_m + 1/2}{N_m + 1} - \t_1\right|> 1/8 \Big| N_m\right) &\le 2\sqrt{e} e^{-N_m /32} .
\end{align}
Next, since $N_m \sim \Bin(n, \Prob(X = m))$ and $\Prob(X = m) = \frac{1}{m}$ by our choice of $P_X$, recalling the moment-generating function of the binomial random variable $\Expt[e^{tN_m}] = \left(1-\frac{1}{m} + \frac{1}{m} e^t\right)^{n}$, we have 
\begin{align}
    \Expt_{N_m}\left[ \Prob\left(\left|\frac{K_m + 1/2}{N_m + 1} - \t_1\right|> 1/8\Big| N_m\right) \right] &\le \Expt_{N_m}[2 \sqrt{e} e^{-N_m/32}] \nonumber \\
    &\le 2\sqrt{e} \left(1-\frac{1}{m} + \frac{1}{m} e^{-1/32}\right)^{n} \label{eq:monotonicFn}.
\end{align}
 We can use the exact same procedure to establish 
\begin{align}\label{eq:ConcIneqsForPHat}
\Prob\left(|\ph_l - \t_{g(l)}|> 1/8 \right) \le 2\sqrt{e} \left(1-\frac{1}{m} + \frac{1}{m} e^{-1/32}\right)^{n}
\end{align}
for $l = 1,\dotsc,m-1$. Substituting this bound into~\eqref{eq:UnionBd} yields 
\begin{align}
    \Prob(\Gh \neq g) &\le 2m\sqrt{e} \left(1-\frac{1}{m} + \frac{1}{m} e^{-1/32}\right)^{n} \nonumber \\
    &= 2m\sqrt{e} \left(1-\frac{(1-e^{-1/32})}{m}\right)^{n} \nonumber \\
    &\le 2\sqrt{e}m e^{-n(1-e^{-1/32})/m} \nonumber \\
    &\le 2\sqrt{e}m  e^{-3n/100 m} 
\end{align}
as required.



 \bibliographystyle{IEEEtran}
 \bibliography{UnivComp.bib}

\begin{thebibliography}{10}
\providecommand{\url}[1]{#1}
\csname url@samestyle\endcsname
\providecommand{\newblock}{\relax}
\providecommand{\bibinfo}[2]{#2}
\providecommand{\BIBentrySTDinterwordspacing}{\spaceskip=0pt\relax}
\providecommand{\BIBentryALTinterwordstretchfactor}{4}
\providecommand{\BIBentryALTinterwordspacing}{\spaceskip=\fontdimen2\font plus
\BIBentryALTinterwordstretchfactor\fontdimen3\font minus
  \fontdimen4\font\relax}
\providecommand{\BIBforeignlanguage}[2]{{%
\expandafter\ifx\csname l@#1\endcsname\relax
\typeout{** WARNING: IEEEtran.bst: No hyphenation pattern has been}%
\typeout{** loaded for the language `#1'. Using the pattern for}%
\typeout{** the default language instead.}%
\else
\language=\csname l@#1\endcsname
\fi
#2}}
\providecommand{\BIBdecl}{\relax}
\BIBdecl

\bibitem{Fogel--Feder17}
Y.~Fogel and M.~Feder, ``On the problem of on-line learning with log-loss,'' in
  \emph{2017 IEEE International Symposium on Information Theory (ISIT)}.\hskip
  1em plus 0.5em minus 0.4em\relax IEEE, 2017, pp. 2995--2999.

\bibitem{Feder--Merhav98}
N.~Merhav and M.~Feder, ``Universal prediction,'' \emph{IEEE Transactions on
  Information Theory}, vol.~44, no.~6, pp. 2124--2147, 1998.

\bibitem{Rissanen83a}
J.~Rissanen, ``A universal data compression system,'' \emph{IEEE Transactions
  on information theory}, vol.~29, no.~5, pp. 656--664, 1983.

\bibitem{Rissanen83b}
------, ``A universal prior for integers and estimation by minimum description
  length,'' \emph{The Annals of statistics}, pp. 416--431, 1983.

\bibitem{Rissanen84}
------, ``Universal coding, information, prediction, and estimation,''
  \emph{IEEE Transactions on Information theory}, vol.~30, no.~4, pp. 629--636,
  1984.

\bibitem{Xie--Barron97}
Q.~Xie and A.~R. Barron, ``Minimax redundancy for the class of memoryless
  sources,'' \emph{IEEE Transactions on Information Theory}, vol.~43, no.~2,
  pp. 646--657, 1997.

\bibitem{Xie--Barron00}
------, ``Asymptotic minimax regret for data compression, gambling, and
  prediction,'' \emph{IEEE Transactions on Information Theory}, vol.~46, no.~2,
  pp. 431--445, 2000.

\bibitem{Shkel--Raginsky--Verdu18}
Y.~Shkel, M.~Raginsky, and S.~Verd{\'u}, ``Sequential prediction with coded
  side information under logarithmic loss,'' in \emph{Algorithmic Learning
  Theory}, 2018, pp. 753--769.

\bibitem{shalev2014understanding}
S.~Shalev-Shwartz and S.~Ben-David, \emph{Understanding machine learning: From
  theory to algorithms}.\hskip 1em plus 0.5em minus 0.4em\relax Cambridge
  university press, 2014.

\bibitem{Lazaric--Munos12}
A.~Lazaric and R.~Munos, ``Learning with stochastic inputs and adversarial
  outputs,'' \emph{Journal of Computer and System Sciences}, vol.~78, no.~5,
  pp. 1516--1537, 2012.

\bibitem{RSTMartingaleLLN15}
A.~Rakhlin, K.~Sridharan, and A.~Tewari, ``Sequential complexities and uniform
  martingale laws of large numbers,'' \emph{Probability Theory and Related
  Fields}, vol. 161, no. 1-2, pp. 111--153, 2015.

\bibitem{rstOLSeqComp15}
------, ``Online learning via sequential complexities.'' \emph{J. Mach. Learn.
  Res.}, vol.~16, no.~1, pp. 155--186, 2015.

\bibitem{RSLogLoss15}
A.~Rakhlin and K.~Sridharan, ``Sequential probability assignment with binary
  alphabets and large classes of experts,'' \emph{arXiv preprint
  arXiv:1501.07340}, 2015.

\bibitem{Bilodeau--Foster--Roy20}
B.~Bilodeau, D.~Foster, and D.~Roy, ``Tight bounds on minimax regret under
  logarithmic loss via self-concordance,'' in \emph{International Conference on
  Machine Learning}.\hskip 1em plus 0.5em minus 0.4em\relax PMLR, 2020, pp.
  919--929.

\bibitem{krichevsky1981performance}
R.~Krichevsky and V.~Trofimov, ``The performance of universal encoding,''
  \emph{IEEE Transactions on Information Theory}, vol.~27, no.~2, pp. 199--207,
  1981.

\bibitem{Yang--Barron99}
Y.~Yang and A.~Barron, ``Information-theoretic determination of minimax rates
  of convergence,'' \emph{Annals of Statistics}, pp. 1564--1599, 1999.

\bibitem{Boucheron--Lugosi--Massart13}
S.~Boucheron, G.~Lugosi, and P.~Massart, \emph{Concentration inequalities: A
  nonasymptotic theory of independence}.\hskip 1em plus 0.5em minus 0.4em\relax
  Oxford university press, 2013.

\bibitem{merhav1995strong}
N.~Merhav and M.~Feder, ``A strong version of the redundancy-capacity theorem
  of universal coding,'' \emph{IEEE Transactions on Information Theory},
  vol.~41, no.~3, pp. 714--722, 1995.

\bibitem{vershynin2018high}
R.~Vershynin, \emph{High-dimensional probability: An introduction with
  applications in data science}.\hskip 1em plus 0.5em minus 0.4em\relax
  Cambridge university press, 2018.

\bibitem{li2011concise}
S.~Li, ``Concise formulas for the area and volume of a hyperspherical cap,''
  \emph{Asian Journal of Mathematics and Statistics}, vol.~4, no.~1, pp.
  66--70, 2011.

\bibitem{10.5555/1146355}
T.~M. Cover and J.~A. Thomas, \emph{Elements of Information Theory (Wiley
  Series in Telecommunications and Signal Processing)}.\hskip 1em plus 0.5em
  minus 0.4em\relax USA: Wiley-Interscience, 2006.

\end{thebibliography}

\end{document}